\documentclass[11pt]{article} 
\usepackage{amsthm,graphicx,amsmath,amssymb,amsfonts,framed}
\usepackage{fullpage}
\usepackage{slashed}
\usepackage{afterpage}
\usepackage{cite}
\usepackage{hyperref}

\newtheorem{prop}{Proposition}[section]

\newtheorem{lemma}[prop]{Lemma}
\newtheorem{cor}[prop]{Corollary}
\newtheorem{df}[prop]{Definition}
\newtheorem{rem}[prop]{Remark}

\numberwithin{equation}{section}

\newcommand{\A}{{\mathcal A}}
\newcommand{\B}{{\mathcal B}}
\newcommand{\D}{{\mathcal D}}
\newcommand{\J}{{\mathcal J}}
\newcommand{\cL}{{\mathcal L}}
\newcommand{\M}{{\mathcal M}}
\newcommand{\HH}{{\mathcal H}}

\newcommand{\I}{{\mathbb I}}
\newcommand{\IM}{{\mathbb M}}
\newcommand{\N}{{\mathbb N}}
\newcommand{\R}{{\mathbb R}}
\newcommand{\C}{\mathbb C}

\newcommand{\cinf}{{C^\infty(\M)}}
\newcommand{\ds}{{\slash \!\!\!\partial}}
\newcommand{\ii}{{\rm i}}
\newcommand{\ncint}{{\int \!\!\!\!\!\! -}}

\DeclareMathOperator{\Aut}{Aut}
\DeclareMathOperator{\id}{id}
\DeclareMathOperator{\Tr}{Tr}

\begin{document}

\title{On twisting real spectral triples by algebra automorphisms} 
\author{Giovanni Landi, Pierre Martinetti} 
%\address{University of Trieste}
\date{}

\maketitle

\begin{abstract}
We systematically investigate ways to twist a real
spectral triple via an algebra automorphism and in particular, we naturally define a twisted partner for any real graded spectral triple. 
Among other things we investigate consequences of the twisting on the fluctuations of the metric and 
possible applications to the spectral approach to the Standard Model of particle physics.
\end{abstract}

\tableofcontents

\parskip=1ex

\vfill
\thanks{\hspace*{-\parindent}
\date{v1: 1st January  2016; v2: August 2016} \\[2pt]
-------- -------- -------- -------- -------- 
\\[5pt]
\textit{Thanks.}
This work was partially supported by the Italian Project ``Prin 2010-11 -- Operator Algebras, Noncommutative Geometry and Applications''. 
\\[10pt]
Giovanni Landi, Universit\`{a} di Trieste, Trieste, Italy and I.N.F.N. Sezione di Trieste, Trieste, Italy. 
\noindent
Pierre Martinetti, Universit\`{a} di Trieste, Trieste and Universit\`{a} di Genova, Genova, Italy. \\
emails: landi@units.it , martinetti@dima.unige.it
}
 
\newpage
\section{Introduction}

We investigate in a systematic way how to twist a
spectral triple, and in particular the consequences of the twisting on the fluctuations of the metric. 
Twisted spectral triples have been defined by Connes and Moscovici in
\cite{Connes:1938fk}. They consist in replacing in the definition of a
spectral triple $(\A, \HH, D)$ the condition that $[D,a]$ be bounded for any
$a\in\A$ by the following:
there exists an automorphism $\rho$ of $\A$
such that the operator which is bounded, for any $a\in\A$, is rather the twisted commutator
\begin{equation} %\label{eq:12}
  [D, a]_\rho := Da - \rho(a) D \, .
\end{equation}

The original motivation of \cite{Connes:1938fk} was to deal with type III operator algebras, for which there is no non trivial trace. 
The examples there were spectral triples perturbed by a conformal transformation and 
spectral triples associated to codimension 1 foliations. 
Twisted spectral triples are relevant for quantum groups (and related spaces) where twisting of the algebra is a natural phenomenon \cite{KMT03}, \cite{HK06}; see \cite{KS12} for a twisted spectral triple 
for the quantum group $SU(2)$. They also appear in 
$C^*$-dynamical systems \cite{Fathizadeh:2015aa}. Recently,  twisted spectral triples
have also occurred in the description of the Standard Model of elementary particles \cite{buckley}. Here twisting allows one to build models beyond the (spectral approach to the)
Standard Model without modifying the fermionic
content of the theory\cite{Devastato:2013fk}, \cite{Chamseddine:2013uq}. This is obtained by twisting the spectral
triple of the Standard Model of \cite{Chamseddine:2007oz} while keeping the
Hilbert space and the Dirac operator untouched.

In the following we generalize this construction to arbitrary spectral
triples. We
first show in Sect.~\ref{sec:realtwist} how to incorporate the real structure
in the twisted framework (Definition \ref{def:realtwist}),  in a way compatible with the fluctuation of
the metric (Proposition \ref{prop:twistfluct}). In Sect.~\ref{sect:minimal}
we formalize the idea of \emph{minimal twist}, that is twisting a
spectral triple without touching the Hilbert space and Dirac
operators (Definition \ref{deftwist}). A procedure to minimally twist any graded spectral
triple is presented in Proposition \ref{proptwist}, extended to the real
case in Proposition \ref{prop:twistreal}. Next, Sect.~\ref{sec:app}
deals with commutative and 
almost commutative geometries with a twisting by grading that is essentially unique. 
Finally, Sect.~\ref{sec:applications} is devoted to some applications, notably to study twisted fluctuations of a free 
Dirac operator and touches on possible uses in the spectral action approach to the Standard Model with a more 
thorough analysis of these reported elsewhere.

%In all the paper $\M$ is a closed spin manifold and all the algebras are unital and complex, except in section
%  \ref{sec:applications} where we consider real algebras. 
%  By default we always assume that a representation is non-degenerate.

\bigskip
\noindent
{\em Acknowledgments}. We thank Paolo Antonini, Ludwik Dabrowski, Gianfausto Dell'Antonio, Koen van den Dungen, Alessandro Michelangeli for useful discussions during the seminars in Trieste on 13th and 19th November  2015 where this work was presented. 

\section{Real twisted spectral triple structure}\label{sec:realtwist} 

We first extend the twisting of spectral triples to include the real structure 
and then introduce twisted-fluctuations of the metric. Proposition \ref{prop:twistfluct} shows that the picture is coherent: a twisted-fluctuated real spectral triple is a real twisted spectral triple.
 
\subsection{Really twisting}
\label{subsec:realtwist}
Recall \cite{Connes-Marcolli} that a spectral triple $(\A, \HH, D)$
consists in an involutive algebra $\A$ 
acting faithfully 
\footnote{When possible we omit the representation symbol and identify 
$a\in\A$ with its representation $\pi(a)\in\cL(\HH)$.}
by bounded operators on a Hilbert space $\HH$ 
together with a self-adjoint operator $D$ with compact resolvent such that 
$[D,a]$ is bounded for any
$a\in\A$. It is graded (or even) when there exists a grading of $\HH$, that
is a self-adjoint operator $\Gamma$ of square $\I$, that commutes with
$\A$ and anticommutes with $D$. 
Furthermore \cite{Connes:1996fu}, a real spectral triple of $KO$-dimension $k\in\{0,1, \dots, 7\}$ modulo $8$,
is a (graded) spectral triple 
together with an antilinear isometry operator $J$ on $\HH$ such that
\begin{equation}
  \label{eq:34}
  J^2= \epsilon(k), \quad JD = \epsilon'(k) DJ, \qquad \mbox{and} \quad J\Gamma = \epsilon''(k) \Gamma J ,
\end{equation}
 where $\epsilon, \epsilon', \epsilon''$ take value in $\left\{-1,+1\right\}$ as a function of $k$ 
 (the explicit table of these signs is not needed in the present paper). 
 Furthermore, the conjugate action of $J$,
   \begin{equation}
     \label{eq:55}
     b \mapsto J b^* J^{-1} \quad \forall b\in\A
   \end{equation}
implements an action of the opposite algebra $\A^\circ$, which is required to commute with the algebra,  
 \begin{equation}
   \label{eq:37}
   [a, J b^* J^{-1}] = 0 \quad \forall a,b \in\A, \qquad \mbox{(zero-order condition)}
 \end{equation}
as well as to commute with the commutator of $D$ with $\A$, 
\begin{equation}
  \label{eq:38}
  [[D, a], Jb^* J^{-1}] = 0 \quad \forall a, b \in \A, \qquad \mbox{(first-order condition)} .
\end{equation}
To avoid ambiguity it may be wise occasionally to reintroduce the representation symbol. Thus, if $\pi$ is the representation of $\A$ on $\HH$, then one gets a representation of $\A^\circ$ on $\HH$ by 
 \begin{equation} \label{oppalg}
\pi^\circ(b):= J \pi(b^*) J^{-1}
 \end{equation}
and \eqref{eq:37} is the statement that the operator algebras $\pi(\A)$ and $\pi^\circ(\A^\circ)$ commute. On the other hand, dropping the representation symbols, we shall write the above as $b^\circ = J  b^* J^{-1}$.

Twisted and graded twisted spectral triples were defined  in \cite{Connes:1938fk} by replacing the boundedness 
of the commutator $[D,a]$ with the requirement that the twisted commutator  
\begin{equation} \label{eq:12}
[D, a]_\rho := Da - \rho(a) D , 
\end{equation}
for an automorphism $\rho\in \Aut(\A)$, be bounded for any $a\in\A$.
Furthermore, the automorphism $\rho$ is not taken to be a $*$-automorphism, but rather to satisfy 
\begin{equation}\label{star-1}
\rho(a^*) = (\rho^{-1}(a))^*. 
\end{equation}
Such an automorphism was named {\em regular} in \cite{KMT03}.
The requirement \eqref{star-1} has origin in the additional assumption (coming from considerations in index theory in \cite{Connes:1938fk}) that the algebra $\A$ has a 1-parameter group of automorphisms $\{\rho_t \}_{t \in \R}$ and that $\rho$ coincides with 
the value at $t=\ii$ of the analytic extension of $\{\rho_t \}_{t \in \R}$. In typical examples (for instance 
the spectral triples associated to codimension 1 foliations) the 1-parameter group of automorphisms is the modular automorphism group of a twisted trace. Such twisted traces appear naturally with twisted spectral 
triples. Indeed, if $(\A, \HH, D)$ is a $\rho$-twisted spectral triple with $D^{-1} \in \cL^{n,\infty}$,
the Dixmier ideal, from \cite[Prop.~3.3]{Connes:1938fk} the functional 
 \begin{equation}
\A \ni a \mapsto \varphi(a) = \ncint \, a D^{-n} := \Tr_\omega(a D^{-n}), 
\end{equation}
with $ \Tr_\omega$ the Dixmier trace, 
is a $\rho^{-n}$-trace, that is $\varphi(a b) = \varphi(b \rho^{-n}(a))$ for all $a,b\in \A$.

%This yields a twisted version of the usual formula for the adjoint of a commutator,
%\begin{equation}
%[D, a]_\rho^* = -[D, \rho^{-1}(a^*)]_\rho.
%\end{equation}

The algebras $\A$ and $\A^\circ$ have isomorphic automorphism groups. 
An isomorphism is:
\begin{equation} 
\Aut(\A) \ni \rho \to \rho^\circ\in \Aut(\A^\circ) , \qquad  
\rho^\circ(b^\circ) := (\rho^{-1}(b))^\circ , \quad \forall b^\circ\in\A^\circ .   
\end{equation}
The use of $\rho^{-1}$ instead of $\rho$ is to parallel condition \eqref{star-1}. In a sense, the above means 
\begin{equation}
  \rho^\circ (J b^* J^{-1}) = J (\rho^{-1}(b))^* J^{-1} = J \rho(b^*) J^{-1},  
\end{equation}
and the second equality is due to condition \eqref{star-1}. 
We are then led to the following.
\begin{df}
\label{def:realtwist}
 A real twisted spectral triple of $KO$-dimension $k$ is the datum of a twisted spectral triple $(\A, \HH, D;\, \rho)$ together with an antilinear isometry  operator $J$ satisfying the rule of signs \eqref{eq:34},  the zero-order condition \eqref{eq:37}, and the
  twisted first-order condition
  \begin{equation}
    \label{eq:39}
      [[D, a]_\rho,  Jb^* J^{-1}]_{\rho^\circ} = 0 , \quad \forall a, b \in \A.
  \end{equation}
\end{df}
\noindent
By inserting the representation symbols and with condition \eqref{star-1}, the above reads as
  \begin{equation} \label{eq:39-bis}
 \big( D\pi(a) - \pi(\rho(a)) D \big)  J \pi(b^*) J^{-1} - J \pi(\rho(b^*)) J^{-1}  \big( D\pi(a) - \pi(\rho(a)) D \big) = 0 , 
 \quad \forall a, b \in \A.
  \end{equation}
We notice that the condition \eqref{eq:39} is symmetric in $\A$ and $\A^\circ$. Indeed, a use of the zero-order conditions $[a, J b^* J^{-1}] = 0$ and $[\rho(a), J (\rho^{-1}(b))^* J^{-1}] = 0$, transforms \eqref{eq:39} into
   \begin{equation}
    \label{eq:39b}
      [[D, Jb^* J^{-1}]_{\rho^\circ},  a]_\rho = 0 , \quad \forall a, b \in \A , 
  \end{equation}
or, for all $a, b \in \A$, 
  \begin{equation} \label{eq:39b-bis}
 \big( D J \pi(b^*) J^{-1} - J \pi(\rho(b^*)) J^{-1} D \big) \pi(a) - \pi(\rho(a))  \big( D J \pi(b^*) J^{-1} - J \pi(\rho(b^*)) J^{-1} D \big) = 0.
  \end{equation}

\begin{rem}
\textup{
One could consider twisting also the zero-order condition
\eqref{eq:37}, and examples from quantum groups 
(see for instance \cite{DLPS}) ---
for which the zero-order condition is valid only modulo infinitesimals
of arbitrary high order --- seems to suggest this possibility. However, from the point of view of the present paper this would introduce unnecessary complication: after all the twist seems to be relevant when the commutator with the operator $D$ is involved. A further, {\it a posteriori}  justification comes from the fluctuation of the metric, as explained
below after Lemma \ref{tformcomm}. % and Proposition \ref{prop:twistfluct}. 
}
\end{rem}

\subsection{Twisted-fluctuation of the metric}
\label{subsec:twistfluc}

Fluctuations of the metric \cite{Connes:1996fu}  easily adapt to the
twisted case.  Given a twisted spectral triple $(\A, \HH, D; \rho)$, one defines
  \begin{equation}
    \label{eq:83}
    \Omega_D^1:=\left\{ \sum\nolimits_j a_j [D, b_j]_\rho\; , \quad a_j, b_j \in \A\right\}
  \end{equation}
the set of twisted $1$-forms. Noticing that 
\begin{equation}
  \label{eq:98}
  [D, ab]_\rho = [D, a]_\rho b + \rho(a) [D, b]_\rho,
\end{equation}
one has \cite[Prop. 3.4]{Connes:1938fk} that $[D,\cdot]_\rho$ is a derivation of $\A$ in $\Omega^1_D$
as soon as the latter is viewed as a $\A$-bimodule with twisted action
on the left:
\begin{equation}
  \label{eq:99}
  a\cdot \xi \cdot b = \rho(a) \,\xi\, b \quad \forall a,b\in\A, \; \xi\in \Omega^1_D.
\end{equation}
\begin{lemma}\label{tformcomm}
For any $A_\rho\in \Omega_D^1$ and any $a, b\in\A$,  it holds that 
\begin{equation}
  [A_\rho , J b^* J^{-1}]_{\rho^\circ} = 0 \quad \mbox{and} \quad [J A_\rho J^{-1}, a]_\rho = 0.
\end{equation}
\end{lemma}
\begin{proof}
 If $A_\rho = \sum_j a_j [D, c_j]_\rho$, for $a_j, c_j \in \A$, by linearity, one needs to show that 
$$
 a_j  [D, c_j]_\rho \, J b^* J^{-1} - J  \rho(b^*) J^{-1}  \, a_j  [D, c_j]_\rho = 0 . 
$$
The zero-order condition \eqref{eq:37} yields $J \rho(b^*) J^{-1}  a_j =   a_j  J \rho(b^*) J^{-1}$ and the l.h.s. becomes
$$
 a_j \left( \, [D, c_j]_\rho \, J b^* J^{-1} - J \rho(b^*) J^{-1} \, [D, c_j]_\rho \, \right)
$$
whose vanishing follows from the twisted first-order condition \eqref{eq:39}. Next, by expanding and inserting 
$J^2$ and $J^{-2}$ (and using $\epsilon^2 =1$ from the signs \eqref{eq:34}) one computes, 
\begin{align*}
0 & =  A_\rho J b^* J^{-1} - J \rho(b^*) J^{-1} A_\rho 
= J^2 A_\rho J^{-2} J b^* J^{-1} - J \rho(b^*) J^{-1} J^2 A_\rho J^{-2} \\
& = J \left( J A_\rho J^{-1} b^* -  \rho(b^*) J A_\rho J^{-1} \right) J^{-1} = J \left( [J A_\rho J^{-1}, b^*]_\rho \right) J^{-1} 
\end{align*}
and renaming $b^*=a$ we get the second equation above.  
\end{proof}
\begin{rem}
\textup{
We see from the above proof that a twisted first-order condition goes well with 
a zero-order condition which is not twisted. 
It is also worth pointing out that, as one would expect, a twisted and an untwisted zero-order condition 
cannot co-exist. By requiring that 
\begin{equation}
[a, Jb^*J^{-1}] = 0 = [a, J b^* J^{-1}]_{\rho^0}      \label{eq:43}  
\end{equation}
for any $a, b \in \A$, a direct computation yields $ J (b^* - \rho(b^*))J^{-1}=0$, that is, $\rho$ has to be the identity. 
On the other hand, as shown by examples below, for finite matrix geometries a twisted and an untwisted first-order condition are not mutually exclusive. 
} 
\end{rem}
\begin{df}
\label{def:twistfluct}  Let $(\A, \HH, D; \rho), J$ be a real twisted spectral triple.
A twisted-fluctuation of $D$ by $\A$ is any
  self-adjoint operator of the kind
  \begin{equation}
    \label{eq:5}
    D_{A_\rho}:= D + A_\rho + \epsilon'  J A_\rho J^{-1}
  \end{equation}
  where $A_\rho\in\Omega_D^1$ and the sign $\epsilon'$ is given as in \eqref{eq:34}.
\end{df}

\noindent Notice that we ask $D_{A_\rho}$ to be self-adjoint, but this is not necessarily the case for $A_\rho$. 

\begin{prop}
\label{prop:twistfluct}
Any twisted-fluctuation $D_{A_\rho}$ of a real twisted
spectral triple $(\A, \HH, D; \rho)$ yields a real twisted
spectral triple
\begin{equation}
  \label{eq:60}
  (\A, \HH, D_{A_\rho}; \rho)
\end{equation}
with the same real structure and $KO$-dimension, and same grading $\Gamma$ (if any).
\end{prop}
\begin{proof} For any $a\in \A$, one has
  \begin{equation}
    \label{eq:54}
    [D_{A_\rho}, a]_\rho  = [D, a]_\rho + [A_\rho, a]_\rho + \epsilon' [ J A_\rho J^{-1}, a]_\rho.
  \end{equation}
  The first term in the r.h.s. is bounded since $(\A, \HH,
  D; \rho)$ is a twisted spectral triple. For the same reason,
  $A_\rho$ is bounded,  being the
  finite sum of products of bounded operators. Thus the second  
  term in the r.h.s. of \eqref{eq:54} is bounded as well, being the
twisted commutator of bounded operators. 
From Lemma \ref{tformcomm} the last term in \eqref{eq:54} vanishes.
Hence \eqref{eq:60} is a twisted spectral triple. It is graded with the
 same grading $\Gamma$ as $(\A, \HH, D; \rho)$ if the latter is graded:  
 one easily checks that $\Gamma$ anticommutes with $A_\rho$ and $J A_\rho J^{-1}$, 
 hence with $D_{A_\rho}$. 

To show that the real structure $J$ of $(\A, \HH, D; \rho)$ 
is a real structure for \eqref{eq:60} with the same $KO$-dimension,
we first check that
\begin{equation}
JD_{A_\rho} = \epsilon' D_{A_\rho} J
\label{eq:61}
\end{equation}
for the same sign $\epsilon'$ as in $JD = \epsilon' DJ$. This follows from
definition \eqref{eq:5}: 
\begin{align}
  \label{eq:62}
  JD_{A_\rho} J^{-1} &= J D J^{-1} + J A_\rho J^{-1}+ \epsilon'  J^2
  A_\rho J^{-2}, \nonumber \\
&= \epsilon' D + J A_\rho J^{-1}+ \epsilon' A_\rho,  \nonumber  \\
&= \epsilon' (D + \epsilon' J A_\rho J^{-1}+ A_\rho) = \epsilon' D_{A_\rho}
\end{align}
 where we used 
 ${\epsilon'}^2=1$, $J^2 = \epsilon \I$ and $J^{-2} = \epsilon^{-1} \I$. 

Finally we must prove the twisted first-order
condition
\begin{equation}
  \label{eq:64}
  [[D_{A_\rho}, a]_\rho, J b^*J^{-1}]_{\rho_0} = 0 \quad \forall a,b
  \in \A.
\end{equation}
Writing $b^\circ = J b^*J^{-1}$,  the l.h.s. of the equation above is 
\begin{equation}
  \label{eq:65}
   [[D, a]_\rho , b^\circ]_{\rho^\circ} + [[A_\rho, a]_\rho, b^\circ]_{\rho^\circ} +\epsilon' [[J A_\rho J^{-1},a]_\rho,b^\circ]_{\rho^\circ}. 
\end{equation}
The first term vanishes by the twisted first-order condition for $(\A,
\HH, D; \rho)$. Next, if $A_\rho\in\Omega_D^1$, it follows that 
$A_\rho' := A_\rho a - \rho(a) A_\rho$ is in $\Omega_D^1$ as well (recall the bimodule structure \eqref{eq:99}).  Then, Lemma \ref{tformcomm} yields 
\begin{equation}
  \label{eq:66}
 [[A_\rho, a]_\rho, b^\circ]_{\rho^\circ} = [A_\rho a - \rho(a) A_\rho, b^\circ]_{\rho^\circ} 
 = [A_\rho', b^\circ]_{\rho^\circ} = 0 , 
\end{equation}
that is, the second term of \eqref{eq:65} vanishes. 
For the third term, again from Lemma \ref{tformcomm} we know that in fact $[J A_\rho J^{-1},a]_\rho = 0$
and the third term of the r.h.s. of \eqref{eq:65} is zero as well. 
\end{proof}

As in the non-twisted case there is a composition law, that is a twisted
 fluctuation of a twisted fluctuation is a twisted fluctuation of the
 initial spectral triple.
 \begin{prop} 
Let 
   \begin{equation}
     \label{eq:85}
     D_\rho= D + A_\rho +\epsilon' J A_\rho J^{-1} \quad \mbox{with}  \quad A_\rho\in \Omega_D^1 
\end{equation}
be a twisted fluctuation of a real twisted spectral
   triple $(\A, \HH, D; \rho)$, and 
   \begin{equation}
     \label{eq:87}
    D'_\rho = D_\rho + A'_\rho + \epsilon' J A'_\rho J^{-1} \quad \mbox{with} \quad A'_\rho\in 
    \Omega_{D_\rho}^1
   \end{equation}
 be a fluctuation of $(\A, \HH, D_\rho; \rho)$. Then 
 \begin{equation}
   \label{eq:88}
   D'_\rho = D_\rho +
     A''_\rho +\epsilon' J A''_\rho J^{-1} \quad \mbox{with} \quad 
  A''_\rho = A_\rho + A'_\rho \in \Omega_D^1.
\end{equation}
 \end{prop}
 \begin{proof}
We wish to show that $D'_\rho = D_\rho + A'_\rho + \epsilon' J A'_\rho J^{-1} = D + A''_\rho + \epsilon ' J A''_\rho J^{-1}$, with
$ A''_\rho \in \Omega_D^1$. 
Let $A_\rho = \sum_j a_k [D, b_k]_\rho$, and 
$A'_\rho = \sum_j a'_k [D_\rho, b'_k]_\rho$ with $a_k, b_k, a'_k, b'_k\in\A$. Omitting the summation
indices and symbol, one has
\begin{align}
  \label{eq:90}
A'_\rho &=  a'[D + A_\rho +\epsilon'  J A_\rho J^{-1} , b']_\rho,  \nonumber\\
& = a'[D,b']_\rho +a'[A_\rho, b']_\rho + \epsilon' a'[JA_\rho J^{-1},  b' ]_\rho.
\end{align}
The first term is in $\Omega^1_D$. The second as well from the
bimodule structure~\eqref{eq:99}. The last term vanishes by Lemma \ref{tformcomm}. 
Hence $A'_\rho$ is in $\Omega^1_D$, and so is $A''_\rho = A_\rho +
A'_\rho$.
 \end{proof}
In other terms, in contrast with the fluctuations without first order condition developed in 
\cite{Chamseddine:2013fk}, twists do not alter the group structure of the fluctuations of the metric.

\section{Minimal twisting for graded spectral triples} 
\label{sect:minimal}

In this section, we work out a general procedure to twist a (real)
graded spectral triple while keeping the Dirac operator and the Hilbert space unchanged. 
The twisting uses the grading.

\subsection{Minimal twisting}
On a manifold there is no room for a twisting; by this we mean the following. 
Start with the canonical spectral triple of a closed spin manifold $\M$,
\begin{equation}
 (\cinf,\, L^2(\M,S),\, \ds := - \ii \gamma^\mu \nabla_\mu) ,
 \label{eq:1}
 \end{equation} 
where $\cinf$ acts on the Hilbert space $L^2(\M,S)$ of square
integrable spinors by multiplication,
\begin{equation}
(\pi_\M(f)\psi)(x) :=f(x)\psi(x),
\label{eq:151}
\end{equation}
and $\ds$ is the Dirac operator, with $\nabla_\mu = \partial_\mu + \omega_\mu$ 
the covariant derivative in the spin bundle. Then any twisted commutator would be of the form
\begin{equation} \label{eq:13}
[\ds, f]_\rho = -\ii \gamma^\mu (\partial_\mu f ) + ( f-\rho(f) ) \, \ds 
\end{equation}
and it would be bounded for any $f\in\cinf$ if and only if
\begin{equation}
\label{eq:14}
f-\rho(f) = 0 ,
\end{equation}
for any function $f$, which just means that $\rho$ is the identity. 

Equation \eqref{eq:14} follows from the following more general result.
\begin{lemma}
\label{lemmacompact}
  Let $(\A, \HH, D)$ be a spectral triple,  and $\rho$ an automorphism
  of $\A$ such that $(\A, \HH, D; \rho)$ is a twisted spectral
  triple. Then $\pi(a) - \pi(\rho(a))$ is a compact operator  for any $a\in\A$. 
\end{lemma}
\begin{proof}
 By simple algebraic manipulations one gets
\begin{equation} \label{eq:1111}
[D, \pi(a)]_\rho = D\pi(a) - \pi(\rho(a)) D = [D, \pi(a)] - \pi( \rho(a) - a) D . 
\end{equation}
Denote $K:= \pi(\rho(a) - a)$; being the representation of $\rho(a)-a\in\A$, this is a bounded operator. 
By definition of a spectral triple, the commutator $[D, \pi(a)]$ is bounded, 
thus the boundedness of $[D, \pi(a)]_\rho$ implies that $KD=[D,
\pi(a)] -[D, \pi(a)]_\rho$ is
bounded.   Hence for any
$\lambda\in\C$ the operator
\begin{equation}
T_\lambda:= KD- \lambda K= K(D-\lambda\I)\label{eq:30}
\end{equation}
is bounded. Again by definition $D$ has a compact resolvent. Since
compact operators on a Hilbert space form an ideal in the algebra of bounded operators, 
one concludes that for any $\lambda$
in the resolvent set of $D$, the operator
\begin{equation}
T_\lambda(D-\lambda\I)^{-1} = K
\label{eq:31}
\end{equation}
is compact.\end{proof}

\noindent When $\A=C^\infty(\M)$, Lemma~\ref{lemmacompact} implies \eqref{eq:14}. Indeed
there is no non-zero function $f\in\A$ that acts as
a compact operator: the spectrum of a compact operator is discrete,
while the spectrum of $\pi_\M(f)$ is the range of $f$, which is
discrete only if $f$ is constant. But then $\pi_\M(f)$ is a multiple
of the identity, which is not a compact operator.
% while the spectrum 
% by the Fredholm alternative, any non-zero $\lambda$ in the
% spectrum of $\pi_\M(f)$ would be an eigenvalue, in contradiction
% with multiplicative operator having no
% eigenvalue. In case $\pi(\A)$ contains the algebra of compact operators - as in
% the Moyal plane - one may work out an automorphism that satisfies the
% condition of lemma \ref{lemmacompact}. We shall not elaborate on that here.}
\smallskip

A way to modify the canonical spectral triple of a manifold in \eqref{eq:1} to allow for non-trivial twistings consists in modifying 
the Dirac operator, for instance by lifting a conformal transformation like is done in \cite{Connes:1938fk}. 
Having in mind applications to the Standard Model of elementary particles, we aim however at
keeping the Dirac operator and the Hilbert space
unchanged, since they encode the fermionic content of the theory that one does not wish to change. 
Then, the only elements we are allowed to play with are the algebra and/or its representation. 
Modifying only the latter does not help: if instead of the multiplicative representation \eqref{eq:151}
one let $f$ acts as $(f\psi)(x) =f(x)p(x) \psi(x) $ 
with $p$ an operator-valued function --- for instance $p$ could be the constant projection on
a subspace $\HH$ of $L^2(\M, S)$, for a reducible representation ---,
then, the extra term in the twisted commutator as in \eqref{eq:13} that needs to vanish for any $f$ is 
$( f-\rho(f) ) p \, \ds $, and the conclusion does not change. 
% exists and is
% compact. Compact operators being an ideal is boundedexists
% and is compact. 
% \begin{equation}
%   \label{eq:27}
  
% \end{equation}
% be a bounded operator for all $a\in\A$. If the representation $\pi$ is faithful the simplest %\rosso{ (the only ????)} 
% possibility is that $\rho(a) = a$ for all $a\in\A$, that is $\rho$ is
% the identity. We are avoiding `pathological' situations where $\pi(
% \rho(a) - a) D$ vanishes while $\pi( \rho(a) - a)$ does not. 
%\end{proof}

Therefore, in order to twist the spectral triple \eqref{eq:1} in a minimal way, that is keeping both 
$\HH$ and $D$ unchanged, one needs to modify the algebra.
\begin{df}
\label{deftwist}
Let $\B$ be a unital involutive algebra.  
 A minimal twisting of a spectral triple $(\A, \HH, D)$ is a twisted
 spectral triple $(\A\otimes \B, \HH, D;\, \rho)$ where  
 $\rho$ an automorphism of $\A \otimes \B$. In addition, the representation of $\A\otimes \I_\B$ 
   coincides with the initial representation of $\A$, that is
 \begin{equation}
\pi(a\otimes \I_\B) = \pi_0(a)   \quad \forall a\in\A
\label{eq:46}
 \end{equation}
where $\pi_0$ and $\pi$ are the representations for $(\A, \HH, D)$ 
and $(\A\otimes \B, \HH, D;\, \rho)$.
\end{df}
\noindent 
Let us comment on the condition \eqref{eq:46}.
 From  the representation $\pi$ of $\A\otimes \B$,
one inherits two representations  of $\A$ and $\B$ on $\HH$,
\begin{equation}
\pi_\A(a):= \pi(a\otimes\I_\B), \quad\pi_\B(b):= \pi(\I_{\A}\otimes b ).
\label{eq:104}
\end{equation}
To make meaningful that  $(\A\otimes \B, \HH, D; \rho)$
is actually a twist of $(\A, \HH, D)$ and not simply a twisted spectral triple
with the same Hilbert space and Dirac operator, it is natural to
impose a relation between $\pi_\A$ and $\pi_0$. The most
obvious one is \eqref{eq:46}, that is
 \begin{equation}
\pi_\A = \pi_0.
\label{eq:101}
\end{equation}
Without any such requirement, Definition \ref{deftwist} would not be very helpful:
one could call ``twist of $(\A, \HH, D)$''  any
twisted spectral triple $(\B, \HH, D; \rho)$ with representation
$\tilde\pi$, by posing $\pi(a\otimes b):=\tilde\pi(b)$. In that case, instead of \eqref{eq:101}
one would have
\begin{equation}
\pi_\A(a) = \I \quad \forall a\in\A.
\label{eq:100}
\end{equation}
One could imagine some alternative to Definition \ref{deftwist} by
imposing a condition less constraining than  \eqref{eq:101} while more
significant than \eqref{eq:100}.
We shall not explore these possibilities here, also because the requirement
\eqref{eq:101} has the following (easy to establish) consequence
that will be of use later on for the Standard Model twisted spectral triple.

\begin{lemma}
A grading $\Gamma$ of the twisted spectral triple $(\A\otimes \B, \HH, D; \rho)$ is a grading of the spectral triple $(\A, \HH, D)$. On the other hand, 
a grading $\Gamma$ of  $(\A, \HH, D)$ is a grading of $(\A\otimes \B,
\HH, D; \rho)$ if and only if
\begin{equation}
[\Gamma, \pi(\I_A\otimes b)]=0 \quad \forall b\in\B.
\label{eq:204}
\end{equation}
\end{lemma}
\begin{proof}
Since the condition that $\Gamma$ anticommutes with $D$ is not touched, it is only a matter of checking  the commuting of $\Gamma$ with the relevant representation. 
If $\Gamma$ is a grading of $(\A\otimes \B, \HH, D; \rho)$, by definition it commutes with $\pi$, that is 
\begin{equation}
[\Gamma, \pi(A)] =0 \quad \forall A\in \A\otimes\B.
\label{eq:205}
\end{equation}
For $A= a\otimes\I_\B$, this yields
\begin{equation}
  \label{eq:206}
  [\Gamma, \pi_0(a)] =0 \quad \forall a\in \A,
\end{equation}
meaning that $\Gamma$ is also a grading of $(\A, \HH, D)$. On the other hand, for a
grading $\Gamma$ of $(\A, \HH, D)$ to be a grading of 
$(\A\otimes \B, \HH, D; \rho)$ one needs $[\Gamma, \pi(A)]= 0$ for any 
$A=\sum_j a_j\otimes b_j\in\A$. Expanding the commutator, one gets
\begin{align}
  \label{eq:207}
  [\Gamma, \pi(A)] &= \sum\nolimits_j [\Gamma, \pi(a_j\otimes\I_\B)\ \pi(\I_\A\otimes b_j)] \nonumber \\
  &= \sum\nolimits_j \Big( \pi_0(a_j)[\Gamma, \pi(\I_A\otimes b_j)] + 
  [\Gamma, \pi_0(a_j)] \pi(\I_\A\otimes b_j) \Big). 
\end{align}
The second term vanishes being 
$\Gamma$ a grading of $(\A, \HH, D)$. The
vanishing of \eqref{eq:207} thus implies \eqref{eq:204}  (take $a_j=\I_\A$). 
Conversely, \eqref{eq:204} implies the vanishing of \eqref{eq:207}. Hence
the result.
\end{proof}

In addition to the previous result, the requirement \eqref{eq:101} leads to a 
necessary condition for a twisted spectral triple $(\A\otimes \B, \HH, D; \rho)$
  to be a minimal twist of a spectral triple $(\A, \HH, D)$.
  \begin{lemma}
\label{lemma:aut}
    Let $(\A\otimes\B, \HH, D; \rho)$ be a minimal twist of a spectral triple
    $(\A, \HH, D)$. Then 
    \begin{equation}
      \label{eq:175}
      \pi(a\otimes\I_\B - \rho(a\otimes \I_\B)) D
\end{equation}
is a bounded operator for any $a\in\A$, implying that $\pi(a\otimes\I_\B -
\rho(a\otimes \I_\B))$ is a compact operator.    \end{lemma}
\begin{proof}
Equation \eqref{eq:1111} for
$b=\I_\B $ gives
\begin{equation}
      \label{eq:177}
      [D, \pi(a\otimes \I_\B)]_\rho = [D, \pi(a\otimes \I_\B)] -  \pi(
      \rho(a\otimes \I_\B)  - a\otimes \I_B)\,D.
    \end{equation}
The twisted commutator on the l.h.s. is bounded by hypothesis. From \eqref{eq:101}
and \eqref{eq:104}, the commutator on the r.h.s. is $[D, \pi_0(a)]$,
which is also bounded by hypothesis. Hence the first claim of the
lemma. The second claim is proven as in Lemma \ref{lemmacompact}.
\end{proof}
\begin{rem}
\label{rem3.5}
\textup{
A similar conclusion for $\I_\A\otimes b$, namely
  \begin{equation}
    \label{eq:183}
    \pi(\I_\A\otimes b- \rho(\I_\A\otimes b)) D \,\in\, \cL(\HH) \quad
    \forall b\in \B,
  \end{equation}
would follow if 
 $[D, \pi(\I_\A\otimes b)]$ were bounded for any $b$ in $\B$.
 But this is not implied by Definition \ref{deftwist}, as
  illustrated by the twisting of graded spectral triples presented in
  Sect.~\ref{sec-twistgrading}: in \eqref{eq:178} the commutator $[D,
  \pi(\I_\A\otimes b)]$ is unbounded. This means that the
  representation $\pi_\B$ in \eqref{eq:104} cannot serve to build a spectral triple $(\mathcal B, \HH, D)$ whose twist by $\A$ would be $(\A\otimes\B, \HH,
  D;\rho)$.
  }
\end{rem}

We shall say that a minimal twist is trivial whenever $\pi_\B(\B)= \C$ or ---  
assuming $\pi_\B$ is faithful --- when $\B = \C$. 
Condition \eqref{eq:101} then puts a constraint on the type of spectral triples that admit interesting minimal twists: 
the starting representation $\pi_0$ of $\A$ on $\HH$ should be reducible. This comes from the following proposition. 
\begin{prop}
  Let $(\A, \HH, D)$ be a spectral triple with representation $\pi_0$. 
  Assume $\A$ is a $(pre-)$ $C^*$-algebra.
  If $\pi_0$ is irreducible, then any minimal twist is trivial.
\end{prop}
\begin{proof}
Let $(\A\otimes\B, \HH, D; \rho)$ with representation $\pi$, be a minimal twist of $(\A, \HH, D)$.  From \begin{equation}
  \label{eq:102}
  \pi(a\otimes b) = \pi(a\otimes \I_\B) \; \pi(\I_\A\otimes b)  =\pi(\I_\A\otimes b) \,\pi(a\otimes \I_\B).
\end{equation}
one has, denoting with $'$ the commutant in $\HH$,
\begin{equation}
  \label{eq:105}
  \pi_\B(\B)\subset \pi_\A(\A)'.
\end{equation}
If  $\pi_0$ is irreducible then $\pi_0(\A)' = \C
\I$ \cite[Prop. II.6.1.8]{blackadar2006}. Hence the result.
\end{proof}

A minimal twist is not the tensor product of $(\A, \HH, D)$ by a spectral triple for $\B$. 
A way to see this is to notice that the twisted commutator $[D, a\otimes b]_\rho$ is not antisymmetric 
in the exchange of its arguments, and so cannot be written as a usual commutator of $a\otimes b$ 
with some operator $D'$. 
Nevertheless one may argue that a minimal twist is somehow a product of
spectral triples where the commutator is then twisted. 
We shall not elaborate much on this here, but only stress that for this to
happen, one needs that the representation $\pi$ of $\A\otimes \B$ on $\HH$ factorizes as the tensor product
\begin{equation}
  \label{eq:109}
  \pi =\tilde\pi_{\A} \otimes \tilde\pi_\B 
\end{equation}
 of
two representations $\tilde \pi_\A, \tilde\pi_\B$ of $\A, \B$ on Hilbert spaces $\HH_\A, \HH_\B$ 
such that $\HH_\A \otimes \HH_\B =\HH$.  
On the other hand, from \eqref{eq:102} the representation $\pi$ is required to be the product
\begin{equation}
  \label{eq:110}
  \pi = \pi_{\A} \, \pi_\B = \pi_\B \, \pi_A
\end{equation}
of two commuting representations of $\A, \B$ on $\HH$ defined in \eqref{eq:104}. 
There is no reason for \eqref{eq:109} and \eqref{eq:110} to be both true at the same time. 

An example where one gets from a representation $\pi_0$ of $\A$ on $\HH$ 
a representation $\pi$ of $\A\otimes \B$ on the same $\HH$ such that \eqref{eq:109} and \eqref{eq:110} both hold, is when
\begin{equation}
\pi_0 = \tilde\pi_\A
\otimes \I_N
\label{eq:107}
\end{equation}
is the direct sum of $N$ copies of an irreducible representation $\tilde\pi_\A$ of $\A$ on an Hilbert space
$\tilde\HH$, and $\B =\IM_N(\C)$. 
Indeed, in this case $\HH$ decomposes as $\widetilde \HH \otimes \C^N$ so that, denoting $\tilde\pi_\B$ the
irreducible representation of $\IM_N(\C)$ on $\C^N$, the representation of $\A\otimes\B$ on $\HH$
\begin{equation}
  \label{eq:108}
  \pi(a\otimes b) := \tilde\pi_\A(a) \otimes \tilde\pi_\B(b)
\end{equation}
  factorizes as in
\eqref{eq:109}. Equation  \eqref{eq:110} holds since 
%Equation \eqref{eq:46} now follows from \eqref{eq:104}:
\begin{equation}
  \label{eq:103}
  \pi_\A(a) := \pi(a\otimes \I_\B) = \tilde\pi_\A(a) \otimes \I_N \qquad \mbox{and} \qquad
\pi_\B(b) = \pi(\I_{\A}\otimes b) = \I_{\tilde\HH}\otimes \tilde\pi_\B(b).
\end{equation}
% Furthermore, since $\pi_\B(b) = \pi(\I_{\A}\otimes b) = \I_{\tilde\HH}\otimes \tilde\pi_\B(b)$, 
% one also checks that \eqref{eq:110} holds. 

\subsection{Twist by grading}
\label{sec-twistgrading}

It is not difficult to minimally twist a
spectral triple $(\A, \HH, D)$ in the sense
of Definition \ref{deftwist} as soon as the latter
is graded. One simply splits the Hilbert space
according to the eigenspaces of $\Gamma$, 
\begin{equation}
 \HH=\HH_+ \oplus \HH_-,
\label{eq:17}
 \end{equation}
and consider the representation of $\A \otimes \C^2\ni (a,a')$ given by
\begin{equation}
  \label{eq:15}
  \pi(a,a'):= p_+\pi_0(a) + p_-\pi_0(a') =\begin{pmatrix} \pi_+(a) & 0 \\ 0&
      \pi_-(a') \end{pmatrix}
\end{equation}
where 
\begin{equation}
  \label{eq:50}
  p_+:= \tfrac12(\I + \Gamma),\quad  p_-:= \tfrac12(\I -\Gamma) 
\end{equation}
are projection on the eigenspaces of $\Gamma$, while
 \begin{equation}
  \label{eq:16}
  \pi_+ (a) := p_+ \pi_0(a)_{\lvert\HH_+}, \qquad  \pi_- (a) := p_- \pi_0(a) _{\lvert\HH_-}
\end{equation}
are the restrictions on $\HH_\pm$ of the representation of $\A$ on
 $\HH$. 
 
\begin{prop} \label{proptwist} 
Let $(\A, \HH, D), \Gamma$ be a graded spectral triple. Then
\begin{equation}
(\A\, \otimes\, \C^2, \HH, D\,;\, \rho)
\label{eq:58}
\end{equation}
with
representation \eqref{eq:15} and automorphism 
   \begin{equation}
     \label{eq:18}
     \rho(a,a') := (a',a), \quad \forall (a,a')\in\A\otimes \C^2
   \end{equation}
is a minimal twist of $(\A, \HH, D)$ with grading $\Gamma$.  
\end{prop}
\begin{proof}  In agreement with \eqref{eq:46}, one retrieves the initial
  representation of $\A$ on $\HH$ as
  \begin{equation}
    \label{eq:113}
    \pi(a, a) = p_+ {\pi_0}(a) + p_- {\pi_0}(a) =  \begin{pmatrix} \pi_+(a) &0 \\ 0&  \pi_-(a)\end{pmatrix}.
  \end{equation}
Since $D$ anticommutes with $\Gamma$, on ${\HH}_+\oplus \HH_-$ it is of the form 
  \begin{equation}
    \label{eq:19}
    D=\begin{pmatrix} 0 & \D \\ \D^\dagger & 0\end{pmatrix}
  \end{equation}
where $\D$ is the restriction of $D$ to $\HH_-$, with image in
$\HH_+$.  Thus by \eqref{eq:15}
\begin{equation}
  \label{eq:20}
  [D,  \pi(a, a')]_\rho = \begin{pmatrix} 0& \D \pi_-(a') -
      \pi_+(a') \D  \\
        \D^\dagger \pi_+(a) - \pi_-(a) \D^\dagger&0 \end{pmatrix}. 
\end{equation}
The lower-left term in \eqref{eq:20} is the restriction to $\HH_+$ of the usual commutator
\begin{equation}
  \label{eq:21}
  [D, \pi(a, a)]=\begin{pmatrix}  0 &\D  \pi_-(a) -
      \pi_+(a) \D  \\
        \D^\dagger \pi_+(a) - \pi_-(a) \D^\dagger&0\end{pmatrix},
\end{equation}
which is bounded since  $(\A, \HH, D)$ is a spectral triple.  
Similarly, 
the upper-right term in \eqref{eq:20} is bounded, being the restriction of $[D, \pi(a', a')]$ to $\HH_-$. 
Hence \eqref{eq:20} is bounded and thus \eqref{eq:58} is a twisted spectral triple.

Since $[\Gamma, \pi(a, a')]=0$, and $\left\{\Gamma, D\right\}=0$ by hypothesis, the spectral triple \eqref{eq:58} 
is $\Gamma$-graded.
\end{proof}

It is easy to see that the flip automorphism \eqref{eq:18} is implemented on the Hilbert space by exchanging the components 
$\psi_{\pm} \in \HH_{\pm}$, that is, for all $\alpha\in \A\otimes \C^2, $
\begin{equation}
\label{flip}
\pi(\rho(\alpha)) = U_\rho \pi(\rho(\alpha)) U_\rho^*, \quad \mbox{with} 
\quad 
U_\rho \begin{pmatrix}
\psi_{+} \\ \psi_{-}
\end{pmatrix} = \begin{pmatrix}
\psi_{-} \\ \psi_{+}
\end{pmatrix} . 
\end{equation} 
Notice that we do not need to assume that $\dim \HH_{+} = \dim
\HH_{-}$.
% in the finite dimensional case. 
To stress the role of $\rho$, compare the expression of $[D, \pi(a, a')]_\rho$ in \eqref{eq:20} with the usual commutator 
\begin{equation}
  [D, \pi(a, a')]=\begin{pmatrix}  0 &\D  \pi_-(a') - \pi_+(a) \D  \\
        \D^\dagger \pi_+(a) - \pi_-(a') \D^\dagger&0\end{pmatrix}.
\end{equation}
While the boundedness of the twisted commutator $[D, \pi(a, a')]_\rho$ follows from the boundedness of  
$[D, \pi(a,a)]$ and $[D, \pi(a', a')]$, there is no reason for the commutator $[D, \pi(a,a')]$ to be bounded.
This is also true for the commutator of $D$ with the
  representation $\pi_\B$ in \eqref{eq:104},  as pointed out in Remark
  \ref{rem3.5}. For
$b=(z_1, z_2)\in\C^2$, one has 
\begin{equation}
  \label{eq:178}
  [D, \pi_\B(b)] = [D, \pi(\I_\A\otimes b)]= 
\left(\begin{array}{cc}
    z_1\I & 0 \\ 0& z_2\I
\end{array}\right)\, 
\left(\begin{array}{cc}
0 & \mathcal D \\ \mathcal D^\dagger &0\end{array}\right) = \left(\begin{array}{cc}
0 & (z_1-z_2)\mathcal D \\ (z_2-z_1)\mathcal D^\dagger &0\end{array}\right),
\end{equation}
which is bounded if and only if $z_2 = z_1$.

The twist-by-grading of Proposition \ref{proptwist} passes to the real structure.

  \begin{prop} \label{prop:twistreal} 
Let $(\A, \HH, D)$ be a graded real spectral triple, with grading $\Gamma$ and 
real structure $J$. Then the twisted spectral $(\A\otimes\C^2, \HH,
    D;\, \rho)$ of Proposition \ref{proptwist} is a graded real twisted
    spectral triple with the same real structure $J$ and the same $KO$-dimension. 
  \end{prop}
  \begin{proof}
The operators $\Gamma, D$ and $J$ are unchanged by the twisting, so the
$KO$-dimension is not modified by passing from $(\A, \HH, D)$ to  $(\A\otimes\C^2, \HH,
    D;\, \rho)$.  One simply needs to check the zero-order and the twisted first-order
    condition \eqref{eq:39}, for both possibilities that $J$ commutes or anticommutes with $\Gamma$, 
    depending on the $KO$-dimension. Notice that the explicit value of the $KO$-dimension and the additional relations implied by this value do not play any role in the proof. 
    
    Since the automorphism $\rho$ in \eqref{eq:18} is just the flip, 
    one has $\rho^2=\id$ and $\rho$ coincides with its inverse.
Assume first that $J$ commutes with
    $\Gamma$. On $\HH=\HH_+ \oplus \HH_-$ one has
    \begin{equation}
      \label{eq:4101}
      J =\begin{pmatrix} J_+ &0 \\ 0& J_-\end{pmatrix}.
    \end{equation}
    For any $(a, \alpha)$, $(b,\beta)$ in $\A\otimes\C^2$ we write
$A=\pi(a,\alpha)$, $B=\pi(b,\beta)$, that is 
$$  A= \begin{pmatrix} a_+ &0
      \\ 0& \alpha_-\end{pmatrix},\quad
 \rho(A) = \begin{pmatrix} \alpha_+ &0
      \\ 0& a_-\end{pmatrix},\quad
   J B^* J^{-1}= \begin{pmatrix} b_+^\circ &0
      \\ 0& \beta_-^\circ \end{pmatrix},\quad
  J \rho(B^*) J^{-1}= \begin{pmatrix} \beta_+^\circ &0
      \\ 0& b_-^\circ \end{pmatrix}
$$
where
\begin{equation}
a_\pm:=\pi_\pm(a),\quad b_\pm^\circ := J_\pm\pi_\pm(b^*) J_\pm^{-1}\label{eq:84}
\end{equation}
and similarly for $\alpha$ and $\beta$. The zero-order condition amounts to 
\begin{equation}
  \label{eq:106}
  [a_+, b_+^\circ] = [\alpha_-, \beta_-^\circ] =0,
\end{equation}
which follows from the zero-order condition for $(\A, \HH, D)$, namely
\begin{equation}
[(a,a) , (b^\circ, b^\circ)] = [(\alpha, \alpha), (\beta^\circ, \beta^\circ)]= 0.
\label{eq:111}\end{equation}
For the twisted first-order condition, one uses \eqref{eq:20} to get
\begin{align}
  [D, A]_\rho\, JB^*J^{-1} &= 
\begin{pmatrix} 
    0& (\D \alpha_- - \alpha_+ \D)   \beta_-^\circ  \nonumber \\
    (\D^\dagger a_+ -  a_- \D^\dagger)  b_+^\circ&0 
  \end{pmatrix}
\\ ~\\  
J\rho(B^*)J^{-1} [D, A]_\rho& = 
   \begin{pmatrix} 
    0&  \beta_+^\circ  (\D \alpha_- - \alpha_+ \D) \\
    b_-^\circ  (\D^\dagger a_+ -  a_- \D^\dagger)  &0 
  \end{pmatrix} . \nonumber
\end{align}
The lower-left component of
\begin{equation}
\big[ [D,A]_\rho, JB^*J^{-1} \Big]_{\rho^\circ} =  [D, A]_\rho\, JB^*J^{-1}  -
J\rho(B^*)J^{-1}[D, A]_\rho\label{eq:97}
\end{equation} 
using \eqref{eq:21} is the lower-left component of:
\begin{align*}
 & \big[[D, (a, a)], (b^\circ, b^\circ)\big] = 
 \begin{pmatrix}
0 & \hspace{-.75truecm}(\D a_- -a_+\D)b_-^\circ -b_+^\circ(\D
a_- - a_+\D)\\
 (\D^\dagger a_+ -a_-\D^\dagger) b_+^\circ - b_-^\circ(\D^\dagger a_+ - a_-\D^\dagger)  &0
\end{pmatrix}
\end{align*}
which vanishes since $(\A, \HH, D)$ satisfies the first-order
condition. Similarly the upper-right component of \eqref{eq:97}
vanishes, being the upper right component of $\big[ [D,(\alpha,\alpha) \big], (\beta^\circ, \beta^\circ) ]$.
Hence the twisted first-order condition is satisfied.

When $J$ anticommutes with $\Gamma$, one has
\begin{equation}
      \label{eq:41}
      J =\begin{pmatrix} 0&\J \\ \epsilon \,\J^{-1}& 0\end{pmatrix},
    \end{equation}
so that, writing now $b_+^\circ := \J^{-1}\pi_+(b^*)\J$, and $\beta_-^\circ := \J\pi_-(\beta^*)\J^{-1}$,
\begin{equation}
  \label{eq:49}
  J B^* J^{-1}= \begin{pmatrix} \beta_-^\circ &0
      \\ 0& b_+^\circ\end{pmatrix}, \qquad  J \rho(B^*) J^{-1}= \begin{pmatrix} b_-^\circ &0
      \\ 0& \beta_+^\circ \end{pmatrix} .
\end{equation}
%where now $b_+^\circ := \J^{-1}\pi_+(b^*)\J$, and $\beta_-^\circ := \J\pi_-(\beta^*)\J^{-1}$.
The proof is then similar to the previous case.
\end{proof}

Propositions \ref{proptwist} and \ref{prop:twistreal} give a way to minimally twist a (real) graded spectral triple using its grading. This needs not be the only possibility, although this happens to be the case for an
even dimensional manifold, as showed in Proposition \ref{prop-dim-4}.
In particular, while it is important for the construction that the grading $\Gamma$ commutes with the algebra (otherwise there would be no guarantee that the
restrictions $p_\pm \A$ of $\A$ to $\HH_\pm$  are algebra
representations, unless the $p_\pm$ themselves are elements of the algebra),
the condition that $\Gamma$ anticommutes with the Dirac operator may
be slightly relaxed, as illustrated later on in Sect.~\ref{sec-almostcomm}. 
More precisely, given a minimal twist $(\A\otimes \C^2, \HH, D; \rho)$ of a spectral triple $ T:=(\A, \HH, D)$,
the extended representation $\pi$ can alway be written  
%with $\rho$ as in \eqref{eq:18},
as in \eqref{eq:15} with a suitable unique grading of the Hilbert space $\HH$. Indeed, by defining 
  \begin{equation}
    \label{eq:114}
    \widetilde\Gamma :=  \pi(\I_\A\otimes (1,-1)) = \pi(\I_\A, -\I_\A) , 
  \end{equation}
a direct computation leads to 
  \begin{equation}
    \label{eq:211}
    \pi(a,a') = \tfrac12(\I +\widetilde\Gamma)\,\pi_0(a) + \tfrac12(\I
    -\widetilde\Gamma)\,\pi_0(a') \quad\quad \forall (a,a')\in\A\otimes\C^2. 
  \end{equation}
Clearly, the operator $\widetilde\Gamma$ defined in \eqref{eq:114} is a grading of $\HH$, that is
$\widetilde\Gamma^*=\widetilde\Gamma$ and $\widetilde\Gamma^2=\I$. It trivially commutes with the representation $\pi_\B$ in \eqref{eq:104}, in fact the latter can be written as 
\begin{equation}
  \label{eq:116}
  \pi_B(z_1, z_2) = \tfrac12(z_1+z_2) \I + \tfrac12(z_1-z_2)\widetilde\Gamma 
  \quad\quad \forall (z_1, z_2)\in\C^2.
\end{equation}
It also commutes with the representation $\pi_0$, since
  \begin{equation}
    \label{eq:203}
    [\widetilde\Gamma, \pi_0(a)]= [\pi(\I_\A\otimes (1,-1)), \pi(a\otimes\I_\B)]=0 .
  \end{equation}
Thus from \eqref{eq:110} it also commutes with $\pi=\pi_0\pi_\B$.
However, $\widetilde\Gamma$ needs not be a grading of the spectral triple, for $\widetilde\Gamma$
may fail to anticommute with $D$.
 If it does, the twisted spectral triple  $(\A\otimes \C^2, \HH, D; \rho)$ is the ``twist by grading'' of the starting spectral triple $(\A, \HH, D)$, obtained by applying Proposition \ref{proptwist} with
$\Gamma=\widetilde\Gamma$. 
Otherwise the minimal twist does not come from the construction of Proposition \ref{proptwist}. 
We come back to this point for the case of almost commutative geometry later on.
\begin{rem}
  \textup{
  There is in fact a constraint on the anticommutator of
    $\widetilde\Gamma$ with $D$ coming from the boundedness of 
    $[D, a]_\rho$. From \eqref{eq:18} one notices that 
 $\rho(\I_\A\otimes (1,-1)) = -\I_\A\otimes (1,-1)$,
so that
    \begin{align}
      \label{eq:26}
      [D, \pi(\I_\A, 0)]_\rho &= \tfrac 12[D,\, \I +\widetilde\Gamma]_\rho =\tfrac12\left(
      D\,\pi\left(\I_\A\otimes (1,-1)\right) + \pi\left(\I_\A\otimes(1,-1)\right)\,D\right) \nonumber \\  
      & =
    \tfrac 12 \, \big(D \, \widetilde\Gamma + \widetilde\Gamma \, D \big) .
    \end{align}
Hence, in any minimal twist $(\A\otimes \C^2, \HH, D; \rho)$ with
    $\rho$ as in \eqref{eq:18}, the anticommutator 
    $\big\{D, \widetilde\Gamma\big\}$ is a bounded operator.
      }  
\end{rem}

\section{Unicity of the twist}
\label{sec:app}

We show in Sect.~\ref{sub:secevendim} that twisting by grading as described in
Sect.~\ref{sec-twistgrading} is the only way to minimally
twist an even dimensional spin manifold. With some conditions of
irreducibility, the same is true for almost commutative geometries as
soon as one uses the real structure, as shown in Sect.~\ref{sec-almostcomm}.
\subsection{Even dimensional manifold}
\label{sub:secevendim}
Let $\M$ be a closed spin manifold of even dimension $n=2m$, $m \geq 1$. The
Euclidean Dirac matrices $\gamma_{[2m]}$ in the chiral basis are the $p:=2^m$
dimensional square
matrices defined recursively
by
\begin{equation}
  \label{eq:133}
  \gamma_{[2]}^1=\sigma_1\quad \gamma^2_{[2]}= \sigma_2\quad
  \gamma_{(2)} = -\ii  \gamma_{[2]}^1\, \gamma_{[2]}^2 = \sigma_3
\end{equation}
 where  % in which $\sigma^\mu = (\I_2, i\sigma_k), \tilde\sigma^\mu= (\I_2, -i\sigma_k)$ where
 \begin{equation}
  \label{eq:2}
  \sigma_1 = \begin{pmatrix} 0 & 1\\ 1&0 \end{pmatrix},\quad
    \sigma_2= \begin{pmatrix} 0& - \ii \\ \ii & 0 \end{pmatrix},\quad
      \sigma_3 = \begin{pmatrix} 1& 0 \\ 0 & -1\end{pmatrix}
\end{equation}
 are the Pauli matrices, and 
\begin{align}
  \label{eq:115}
  \gamma^k_{[2m+2]} &=
\begin{pmatrix} 
  0_{2m} & \gamma^k_{[2m]}  \nonumber \\ 
  \gamma^k_{[2m]} & 0_{2m}  
\end{pmatrix}\quad \text{ for } k=1, ..., 2m\\[4pt]
 \gamma^{2m+1}_{[2m+2]} &=
 \begin{pmatrix} 
 0_{2m} & \gamma_{(2m)} \\ 
 \gamma_{(2m)} & 0_{2m}  
\end{pmatrix},\quad 
 \gamma^{2m+2}_{[2m+2]} =
 \begin{pmatrix} 
 0_{2m} & -\ii \, \I_{2m}\\ 
 \ii \, \I_{2m} &0_{2m}  \end{pmatrix}
\end{align}
where $\gamma_{(2m)}$ is the grading operator 
\begin{equation}
  \label{eq:82}
\gamma_{(2m)}:= (-\ii)^{m} \, \gamma_{[2m]}^1\,\gamma_{[2m]}^2 \cdots
\gamma_{[2m]}^{2m} =\left(\begin{array}{cc} \I_{2m}& 0_{2m}\\ 0_{2m} & -\I_{2m}\end{array}\right).  
\end{equation}

\begin{lemma}
\label{commutator-Dirac-matrices}
Let $A,B\in \IM_{2^{m+1}}(\C)$, be such that
\begin{equation}
\gamma^\mu_{[2m+2]} \,A= B\,\gamma^\mu_{[2m+2]} \quad\forall\mu=1, ... , m+2.
\label{eq:117}
\end{equation}
Then, there exist $\lambda, \lambda'\in\C$ such that 
\begin{equation}
A= \begin{pmatrix} \lambda\, \I_{2^m} & 0_{2^m}\\[8pt] 0_{2^m}&
    \lambda'\,\I_{2^m} \end{pmatrix}, \qquad B= \begin{pmatrix}
    \lambda'\, \I_{2^m} & 0_{2^m}\\[8 pt] 0_{2^m}&
    \lambda\,\I_{2^m} \end{pmatrix}.
\label{eq:4}
\end{equation}
\end{lemma}
\begin{proof}
Let
  \begin{equation}
    \label{eq:7}
    A=\begin{pmatrix} a & b
        \\ c & d \end{pmatrix},\qquad   
    B=\begin{pmatrix} a' & b'
        \\ c'& d' \end{pmatrix} , 
  \end{equation}
be non zero matrices whose entries are $2^m$-square matrices. 
For $\mu=2m+2$ requiring \eqref{eq:117} implies
\begin{align}
b' = -c , \quad c' = -b , \quad \mbox{and} \quad 
a' = d , \quad d' = a .
\label{eq:119}
\end{align}
Then, for $k=1, ... , 2m+1$, one obtains  
\begin{align}
  \label{eq:8}
\gamma^k_{[2m]} a = a \gamma^k_{[2m]} , \quad
\gamma^k_{[2m]} d = d \gamma^k_{[2m]} \qquad
\mbox {and} \qquad
\gamma^k_{[2m]} c = -c \gamma^k_{[2m]} , \quad
\gamma^k_{[2m]} b = -b \gamma^k_{[2m]}\end{align} 
and similar relations with $\gamma_{(2m)}$. 
Thus $b$ and $c$ should anticommute with all the
$\gamma^k_{[2m]}$ as well as with their product $\gamma_{(2m)}$,
which is not possible, unless $b=c=0$. Meanwhile $a$ and $d$ should commute
with all the $\gamma^k_{[2m]}$, which is possible only if $a$
and $d$ are multiple of the identity. 
%(by construction the center of the Clifford algebra reduces to the multiples of the identity)
Hence
the result.  
\end{proof}

The twist by grading of Sect.~\ref{sec-twistgrading} turns out to be the only way to minimally twist the spectral triple
of a manifold  \eqref{eq:1} by a finite dimensional algebra, provided
the latter acts faithfully.
\begin{prop}
\label{prop-dim-4}
 Let $\M$ be a closed manifold of dimension $2m$; $\B$ be a finite dimensional $C^*$-algebra
 and $\rho$ a non-trivial automorphism of
  $\cinf\otimes \B$ such that
  \begin{equation}
(\cinf\otimes \B,
  L^2(\M,S),\, \ds\,;\; \rho)\label{eq:23}
  \end{equation}
 is a minimal twist of the canonical triple $(\cinf, L^2(\M, S), \ds)$, 
 with $\pi_\B$ as defined in \eqref{eq:104} taken to be faithful. Then $\B=\C^2$
and \begin{equation}
\rho(f,g) = (g,f)  \quad\forall(f,g)\in\cinf\otimes \C ^2.
\label{eq:47}
  \end{equation}
Moreover the
 representation $\pi$ of $\cinf \otimes {\B}$ on $L^2(\M,S)$ is
 given by \eqref{eq:15}-\eqref{eq:16} with $\Gamma=\gamma_{(2m)}$ the grading of the
 canonical spectral triple of $\M$ in \eqref{eq:1}.
\end{prop}
\begin{proof}
Let $\I_\M$ denote the identity of $\cinf$. Any element $\I_\M\otimes b$ acts on $\HH$ as a constant
matrix
\begin{equation}
B:=\pi(\I_\M\otimes b)=\pi_\B(b)\label{eq:138}
\end{equation}
of dimension at most
$2^m$. Thus $\pi_\B(\B)$ is a subalgebra of $\IM_{2^m}(\C)$, and
since $\pi_\B$ is
faithful the same is true for the algebra $\B$.
For any $b\in\B$, one finds for the twisted commutator
%{\footnote{We use $[ab, c]_\rho = a[b,c] + [a, c]_\rho b.$}}
  \begin{equation}
    \label{eq:3}
    [\ds, B]_\rho = - \ii \gamma^\mu[\omega_\mu, B] - \ii [\gamma^\mu,  B]_\rho\nabla_\mu , 
  \end{equation}
using $\ds = -\ii \gamma^\mu \nabla_\mu = -\ii \gamma^\mu(\partial_\mu + \omega_\mu)$. 
By a similar argument as below Lemma \ref{lemmacompact},
this is bounded if and only if 
  \begin{equation} 
  \gamma^\mu B - \rho(B) \gamma^\mu = 0 , \quad \forall \mu = 1, ..., 2m .
    \label{eq:9}
  \end{equation}
Then by Lemma \ref{commutator-Dirac-matrices}, the algebra $\B$ is isomorphic either
  to
the algebra of block-diagonal matrices
\begin{equation}
\text{diag}\,(\lambda\, \I_{2^m}, \lambda'\, \I_{2^m}) , %\simeq \C^2
\label{eq:120}
\end{equation}
with $\rho$
the permutation of the two-blocks, or to a subalgebra of it. The first
case yields $\B \simeq \C^2$ resulting into an automorphism of $ \cinf\otimes \B$ $\rho$ 
as given in \eqref{eq:47}. The second case means $\B=\C $ with $\rho$ the trivial identity automorphism, excluded by hypothesis. 

To establish the last point of the proposition, it suffices 
to show that the operator $\widetilde\Gamma$ defined in \eqref{eq:203}
coincides with the grading $\gamma_{(2m)}$, possibly up to an
irrelevant global sign. From \eqref{eq:120} and
\eqref{eq:138} one indeed gets
\begin{equation}
  \label{eq:59}
  \pi_\B(\lambda_1,\lambda_2) = \pm \text{diag}\,(\lambda_1, \lambda_2)\otimes
  \I_{2^m} \qquad \forall (\lambda_1, \lambda_2)\in\C^2,
\end{equation}
hence $\widetilde\Gamma =\pm \gamma_{(2m)}$ as stated. 
\end{proof}

\subsection{Almost commutative geometries}
\label{sec-almostcomm}

For $\M$ a closed spin manifold of even dimension $2m$, the product of the canonical spectral triple \eqref{eq:1},
with grading $\gamma_{(2m)}$ and real structure $\mathcal J$,  by a finite dimensional
unital spectral triple
$(\A_F, \HH_F, D_F)$ is the spectral triple
\begin{equation}
  \label{eq:28}
  \A=\cinf \otimes \A_F, \quad\HH= L^2(\M,S)\otimes \HH_F,\quad D= \ds\otimes \I_F + \gamma_{(2m)}\otimes D_F
\end{equation}
where $\I_F$ is the identity on $\HH_F$,  and the representation
\begin{equation}
 \pi_0 = \pi_{\M}\otimes \pi_F
\label{eq:96}
 \end{equation}
of $\A$ on $\HH$ is the tensor
product of the multiplicative representations \eqref{eq:151} of $\cinf$ on spinors,
by the representation $\pi_F$ of $\A_F$ on $\HH_F$. In addition, when $(\A_F, \HH_F, D_F)$  has grading $\Gamma_F$ and real
structure $J_F$, then the product $(\A, \HH, D)$ is graded and real with
\begin{equation}
  \label{eq:71}
  \Gamma = \gamma_{(2m)}\otimes \Gamma_F, \quad J = \J\otimes J_F. 
\end{equation}

As for the canonical spectral triple for a manifold, there is no room for twisting the product spectral triple 
\eqref{eq:28} while keeping $\A$, $\HH$ and $D$ unchanged.
Indeed, if such a twist $\rho$ exists then by Lemma~\ref{lemmacompact} one
has that
$ f_\rho:=\pi_0(f\otimes m - \rho(f\otimes m))$ is compact for any $f\in C^\infty(\M)$ and 
 $m\in {\mathbb M}_n(\C)$. The same is true for
 \begin{equation}
\tilde f:= (\I\otimes e_{11}) \,f_\rho\,
 (\I\otimes  e_{11}) =: g\otimes e_{11}
\label{eq:10}
 \end{equation}
where $e_{11}=\text{diag} (1, 0, ..., 0)\in {\mathbb  M}_n(\C)$ and $g=\langle\I\otimes e_{11}, f_\rho(\I\otimes
 e_{11})\rangle\in C^\infty(\M)$.  The spectrum of
 $\tilde f$ is the range of $g$, and $\tilde f$ is never compact for
 the reasons explained below Lemma~\ref{lemmacompact}.

From now on we assume that $\A_F$ is a $C^*$-algebra, which is the
case in all the models of almost-commutative geometries applied to physics. The possibilities to minimally twist an almost commutative geometry are
a bit larger than the ones for manifolds, due to possible degeneracies of the representation of $\A_F$ on
$\HH_F$. Before proving this, let us begin with a lemma showing that the (minimal) twisting automorphism
$\rho$ actually acts only on the extra algebra $\B$. 

\begin{lemma}
\label{lemacg}
  Let   $(\A\otimes \B, \HH,\, D\,;\; \rho)$
 be a non-trivial minimal twist of the spectral triple 
 \eqref{eq:28}
by a finite dimensional $C^*$-algebra $\B$. Then
there exists $\rho'\in \Aut(\A_F\otimes\B)$ such that
\begin{equation}
  \label{eq:139}
  \rho(f\otimes T) =  f\otimes \rho'(T) \quad \forall f\in\cinf,\, T\in\A_F\otimes\B . 
\end{equation}
\end{lemma}
\begin{proof}
By Lemma \ref{lemma:aut}, one has that 
\begin{multline*}
%  \label{eq:179}
  \pi(a\otimes \I_\B - \rho(a\otimes \I_\B) )\,(\ds\otimes \I_F + \gamma^5\otimes D_F) \\ = \pi(a\otimes \I_\B - \rho(a\otimes \I_\B) )
  (\ds\otimes \I_F) +\pi(a\otimes \I_\B - \rho(a\otimes \I_\B) )(\gamma^5\otimes D_F)
\end{multline*}
is bounded for any $a$. The second term in the r.h.s. is always
bounded. 
On the other hand, the image of the faithful representation $\pi$ is made of finite matrices of multiplicative operators. Thus,
%By the same argument as after \eqref{eq:11111}, 
the first term is bounded if and
only if  $\pi(a\otimes \I_{\B}- \rho(a\otimes \I_\B))=0$, that is 
\begin{equation}
  \label{eq:148}
\rho(a\otimes \I_\B) = a\otimes \I_\B \quad \forall a\in\A.  
\end{equation}
Therefore $\rho$ preserve the center $\cinf$ of $\A\otimes \B$ and by \cite{KR12} 
is a function from $\M$ to inner automorphisms of the finite part algebra $\A_F\otimes \B$. 
We next show that for our case this function has to be a constant one. 
Let $k:=\text{dim}\, \HH_F$.  For any $T\in\A_F\otimes \B$,  the element $\I_\M\otimes T$ 
acts on $\HH$ as a constant $2^m k \times 2^m k$ matrix
\begin{equation}
  \label{eq:129}
  M :=\pi\,(\I_\M\otimes T)= \big\{ M_{jl} \big\}_{j, l = 1, ..., k}
\end{equation}
where each block  $M_{jl}$ is in $\IM_{2^m}(\C)$. 
On the other hand, if we write 
\begin{equation}
  \label{eq:1291}
\rho(\I_A\otimes T) = \sum\nolimits_j f^j \otimes T_j 
\end{equation}
for some  $ f^j \in\cinf$ and $T_j\in\A_F\otimes\B$, its representation 
$\pi(\rho(\I_A\otimes T))$ is a matrix $\widetilde{M}:=\big\{\widetilde{M}_{jl}\big\}$ 
where each block $\widetilde{M}_{jl}$ is a priori a function on $\M$, that is an element in 
$C^\infty(\M, M_{2^m}(\C))$.
The operator $\ds\otimes \I_F$ acts
as $\text{diag}\,(\ds, \ds, ..., \ds)$, so that
\begin{equation}
  \label{eq:191}
  (\ds\otimes\I_F)\,\pi(\I_\M\otimes T) - \pi(\rho(\I_\M\otimes T)  
  (\ds\otimes\I_F)= \left\{(\ds M_{jl}) +  (\gamma^\mu M_{jl} - \widetilde{M}_{jl} \gamma^\mu) \partial_\mu\right\}
\end{equation}
is bounded in and only if
\begin{equation}
  \label{eq:192}
  \gamma^\mu M_{jl} = \widetilde{M}_{jl} \gamma^\mu  \quad \forall j,l\in[1,k].
\end{equation}
This means that all the $\widetilde{M}_{jl}$'s are constant or, 
given the nature of the representation $\pi$, that all $f^j$'s in \eqref{eq:1291} are constant  
Therefore \eqref{eq:1291} reads
\begin{equation}
  \label{eq:193}
  \rho(\I_\M\otimes T)
= \I_\M\otimes \rho'(T)
\end{equation}
where the automorphism $\rho'\in \text{Aut}\, (\A_F\otimes \B)$ is
defined by
\begin{equation}
  \label{eq:194}
  \rho'(T):=\sum\nolimits_j  f^{j} T_j . 
\end{equation}
%if  $\rho(\I_A\otimes T) = \sum\nolimits_j f^j \otimes T_j$. 
Using \eqref{eq:148} it is straightforward that 
$\rho(f\otimes T) =  \rho (f\otimes \I_{{\A_F}\otimes\B})\, \rho(\I_\M\otimes T) = f\otimes \rho'(T)$,
for all $f\in\cinf,\, T\in\A_F\otimes\B$, which proves the statement of the lemma.
\end{proof}

Since $\A_F$ is a $C^*$-algebra, it is a sum of matrix algebras,
\begin{equation}
  \label{eq:142}
  \A_F = \underset{i=1}{\overset{q}\bigoplus} \, \IM_{n_j}(\C)    \quad n_j\in\N^*.
\end{equation}
The representation $\pi_F$ is faithful, and we assume that each of the $\IM_{n_j}(\C)$ acts faithfully 
on $\HH_F$ as the direct sum of
$d_j$ copies of the fundamental representation. 
The dimension $k$ of $\HH_F$ is $k=\sum_j n_j d_j$ 
and we denote
\begin{equation}
d:= \min \left\{d_1, d_2, ..., d_q\right\}.\label{eq:145}
\end{equation}
   \begin{prop}
\label{prop-almost}
 Let   
$(\A\otimes \B,
  \HH,\, D\,;\; \rho)$
 be a non-trivial minimal twist of the almost commutative spectral triple 
 \eqref{eq:28} with $\A_F$ as above, and $\B$ a finite dimensional $C^*$-algebra such that
 $\pi_\B$ in \eqref{eq:104} is faithful. 
Then
\begin{equation}
\B=\C^l\otimes \C^2\quad \text{ for some} \quad l\in[1, d], 
\label{eq:146}
\end{equation}
with $d$ defined in \eqref{eq:145} and,  for all $(a_1, ..., b_1, ...) \in\A\otimes {\mathbb C}^2\otimes \C^l$ 
the automorphism is
 \begin{equation}
\rho(a_1, a_2, ..., a_l, b_1, b_2, ..., b_l) = (b_1, b_2, ..., b_l, a_1, a_2, ..., a_l) 
\label{eq:47acg}
  \end{equation}
\end{prop}
\begin{proof} The notations are those of Lemma \ref{lemacg}.
   By Lemma \ref{commutator-Dirac-matrices}, Equation \eqref{eq:192} implies
    that the matrices $M_{jl}$
    and $\widetilde{M}_{jl}$ are of the form
\begin{equation}
M_{jl} = \left(\begin{array}{cc} \alpha_{jl} & 0 \\ 0
    &\beta_{jl}\end{array}\right)\otimes \I_{2^{m-1}},\quad
\widetilde{M}_{jl} = \left(\begin{array}{cc} \beta_{jl} & 0 \\ 0
    &\alpha_{jl}\end{array}\right)\otimes \I_{2^{m-1}\,} , 
    \qquad \mbox{with} \quad \alpha_{jl}, \beta_{jl}\in\C .
\label{eq:131}
  \end{equation}
Let
\begin{equation}
\mathfrak A \simeq \IM_k(\C)\otimes \C^2\label{eq:141}
\end{equation}
be the $C^*$-algebra generated by all the matrices as in \eqref{eq:129}, with
blocks satisfying
\eqref{eq:131}.
 Since $\pi$ is faithful, $\A_F\otimes \B \simeq
 \I_\M\otimes(\A_F\otimes\B)$ is isomorphic to a subalgebra of
 $\mathfrak A$.

If $\A_F$ is a single matrix algebra, then $\A_F=\IM_k(\C)$ since
$\pi_F$ is faithful. By \eqref{eq:141}, one obtains $\B=\C^2$ with $\rho$ given by
\eqref{eq:47acg} and $l= 1$. This is the statement \eqref{eq:146} where $l=d_1=1$
and all the other $d_j$'s vanishing.

Otherwise, with $M_j^{(r)}$ denoting
the $r$-th copy in $\pi_F(\A_F)$ of the fundamental
representation of the matrix algebra $\IM_{n_j}(\C)$, one has
 \begin{equation}
   \label{eq:157}
 \pi_0(\I_\M\otimes \A_F) = \mbox{diag} \left(
     M_1^{(1)}, \dots, M_1^{(d_1)}, M_2^{(1)}, \dots, M_q^{(d_q)} \right)\otimes \I_{2^m}.
\end{equation}
For any $b\in\B$, the operator $\pi(\I_\A\otimes b)$ commutes with
 the operator $\pi(\A\otimes \I_\B)$ hence, by \eqref{eq:103}, with 
 $\pi(\I_\M\otimes \A_F \otimes \I_\B)=\pi_0(\I_\M\otimes \A_F)$. This means
\begin{equation}
\pi(\I_\A\otimes b) = \mbox{diag} \left(  
     \lambda_1^{(1)} \I_{n_1}, \dots, \lambda_1^{(d_1)} \I_{n_1}, 
     \lambda_2^{(1)} \I_{n_2}, \dots, \lambda_q^{(d_q)}\I_{n_q} 
\right)\otimes T
\label{eq:158}
\end{equation}
for $\big\{\lambda_j^{(t)}\big\}\in\C^{d'}$, with $d':=\sum_j d_j$,  and $T$
an arbitrary matrix in $\IM_{2^m}(\C)$.  But 
$\pi(\I_\A\otimes b)$ belonging to $\mathfrak A$ forces $T$ to be of
the form
 \begin{equation}
T = \left(\begin{array}{cc}\alpha &0 \\ 0&\beta\end{array}\right)\otimes \I_{2^{m-1}}\label{eq:161}
\end{equation}
for some $\alpha, \beta\in\C$. Hence $\B$ is isomorphic to a subalgebra of the algebra 
\begin{equation}
\mathfrak B = \C^{d'}\otimes \C^2
\label{eq:159}
\end{equation}
generated by all elements \eqref{eq:158} with \eqref{eq:161}.
The automorphism $\rho$ is defined as in Lemma
\ref{commutator-Dirac-matrices} by the permutation of $\alpha_{jl}$
and $\beta_{jl}$ in \eqref{eq:131}. Thus it acts only on the
$\C^2$ factor of $\mathfrak B$. Since $\rho$ is non trivial by
hypothesis, this forbids to
consider any subalgebra $\C^l\otimes \C$ of \eqref{eq:159}. Hence 
\begin{equation}
  \label{eq:91}
  \B \simeq \C^{l}\otimes \C^2 \quad \text{ for some } l\leq d'.
\end{equation}
Next, for any $b \in \B$, and $S\in\IM_{n_1}(\C)$ viewed as an element
 of $\A_F$, equations \eqref{eq:157}-\eqref{eq:161} lead to 
\begin{align}
  \label{eq:143}
  \pi(\I_\M\otimes S \otimes b) &= \pi_0(\I_\M\otimes S)\,
  \pi(\I_\A\otimes b) \nonumber \\ 
 & = \mbox{diag} \left(
    \left(\begin{array}{cc} M_1 &0 \\ 0&
                                       N_1\end{array}\right) , \dots, 
    \left(\begin{array}{cc} M_{d_1} &0 \\ 0&
                                       N_{d_1}\end{array}\right) , 0 , \dots, 0 \right) .
  \end{align}
Thus $\pi(\I_\M\otimes\A_F\otimes \B)$ contains at most $2d_1$ independent
representations of $\IM_{n_1}(\C)$. So if $l>d_1$, the
representation $\pi$ of $\I_\M\otimes \A_F\otimes (\C^l\otimes \C^2)$
is not faithfull, which is excluded by hypothesis. Therefore
\begin{equation}
  \label{eq:201}
  l\leq d_1.
\end{equation}
The same is true for all the $d_j$'s, hence the result that $d= \min \left\{d_1, d_2, ..., d_q\right\}$.
\end{proof}

Unlike the case of the canonical triple of a manifold, a minimal twist of an almost
commutative geometry is not necessarily by $\C^2$. However,  although the algebra
$\B=\C^l\otimes \C^2$ may be bigger than $\C^2$, the twisting
automorphism $\rho$ always results in permuting the two components of
spinors like in \eqref{flip}. Thus $\rho$ is an automorphism of the $\C^2$ factor of $\B$, which forms the ``irreducible'' part of the twist, in contrast with the $\C^l$ factor which reflects the 
reducibility of the representation $\pi_F$ of the finite dimensional algebra. 
By adding a condition of irreducibility for the finite part representation $\pi_F$ Proposition \ref{prop-almost}  yields the same unicity result as for manifolds.

\begin{cor}
\label{coracg}
 Let $(\A, \HH, D)$ be an almost commutative geometry 
 as in Proposition \ref{prop-almost}, such that the
 representation $\pi_F$ of $\A_F$ is irreducible. Then 
any  non-trivial minimal twist $(\A\otimes \B,
  \HH,\, D\,;\; \rho)$ is by $\B =\C^2$ with automorphism
  $\rho(a, a') = (a', a)$ for any $a,a'\in\A\otimes \C^2$.\end{cor}
\begin{proof}
  This is Proposition \ref{prop-almost} with all the $d_i$'s equal to
  $1$, so that $l=1$ and $\B = \C^2\otimes \C = \C^2$.
\end{proof}

Nevertheless, there is still a degree of freedom in the representation
$\pi_\B$ of the auxiliary algebra $\C^2$, and thus in the grading operator $\widetilde\Gamma$ as defined in
  \eqref{eq:114}. 
 This freedom could lead to twisting of almost commutative geometries which are not of the `the twist by grading' type, in contrast to what happens for manifolds as shown in Proposition \ref{prop-dim-4}. 
Restricting to the irreducible case where all the
  $d_j$'s are equal to $1$, from \eqref{eq:158} one has:
\begin{equation}
  \label{eq:216}
\widetilde\Gamma:=\pi(\I_\A\otimes (1,-1))=   \bigoplus_{j=1}^q T_j 
\end{equation}
where each $T_j$ is one of the two possible representations of
$(1,-1)$ on $L^2(\M, S)$ allowed by \eqref{eq:16}, that is $T_j =
\gamma_{(2m)}$ or $-\gamma_{(2m)}$. In other terms, one has
\begin{equation}
  \label{eq:219}
  \widetilde\Gamma =\gamma_{(2m)}\otimes \widetilde\Gamma_F
\end{equation}
where $\widetilde\Gamma_F$ is a diagonal matrix with entries $\pm 1$. 
As stressed at the end of Sect.~\ref{sec-twistgrading}, 
the point is whether the operator $\widetilde\Gamma$ is a grading
of the twisted almost commutative geometry or not. If yes, the only 
minimal twist of any such a geometry by $\C^2$ would be
by grading as in Sect.~\ref{sec-twistgrading}; otherwise, there would
be alternative ways to minimally twist an almost
commutative geometry by $\C^2$, even in the irreducible case.
\begin{rem}
\textup{
When $q=1$, that is when $\A_F=M_n(\C)$, the operator $\widetilde\Gamma$ is either
  $\gamma_{(2m)}\otimes \I_F$ or $-\gamma_{(2m)}\otimes\I_F$. This it is
  not a grading of the almost commutative geometry since it does not
  anticommutes with $\gamma_{(2m)}\otimes D_F$. This reflects the fact
  that there is no grading for $\A_F=M_n(\C)$ acting
  irreducibly on $\HH_F$, for  the only operator that commutes with
  $\pi_F(\A_F)$ is the identity.  
}
\end{rem}

\noindent 
So far, we are able to answer this question in the real case, adding the assumption that $\widetilde\Gamma$ behaves well with respect to the real structure $J$, that is
\begin{equation}
  \label{eq:220}
  \widetilde\Gamma J = \tilde\epsilon \, J\Gamma \quad \text{ for some }
  \tilde\epsilon =1 \text{ or } -1.
\end{equation}
\begin{prop}
\label{coracg2}
 Let  
$(\A\otimes\C^2, \HH, D; \rho)$ with $\rho$ as in
\eqref{eq:18} be a minimal twist of an almost
 commutative geometry  $(\A, \HH, D)$ with $\A_F$ as in \eqref{eq:142}. Assume in addition
that the twisted spectral triple is real, with real structure
$J$. If \eqref{eq:220} holds true, then  
\begin{equation}
\widetilde\Gamma \, D + D \, \widetilde\Gamma=0,
\label{eq:208}
\end{equation}
meaning that $\widetilde\Gamma$ is a grading of both the starting and the twisted spectral triples.
\end{prop}
\begin{proof}
We only sketch the proof that goes along the lines of the proofs of  
Propositions \ref{prop:twistreal} and \ref{prop:twistreal} since in a sense the present proposition 
goes in the inverse direction of those.
The key is to decompose $\HH=\tilde\HH_+ \oplus \tilde\HH_-$ 
into the eigenbasis of $\widetilde\Gamma$ and then all operators accordingly. 

Firstly,  the boundedness of the twisted commutator
$[D, \pi(a, \alpha)]_\rho$ for any $(a,\alpha)\in\A\otimes
\C^2$, restricts to requiring only the boundedness of 
\begin{equation}
\label{eq:198}
[-\ii \gamma^\mu\partial_\mu\otimes \I_F, \pi(a, \alpha)]_\rho 
\end{equation}
since, the twisted commutators of $\pi(a,\alpha)$ with $\gamma_{(2m)}\otimes D_F$ and 
$-\ii \gamma^\mu\omega_\mu\otimes\I_F$ are trivially bounded. 
That the twisted commutator in \eqref{eq:198} be bounded leads, with a direct computation to 
\begin{equation}
\label{eq:200}
\left( \gamma^\mu \otimes\I_F \right) \, \widetilde\Gamma + \widetilde \Gamma \left( \gamma^\mu \otimes\I_F \right)=0 
\qquad \forall
\mu=1, ..., 2m. 
\end{equation}
This shows that $\widetilde\Gamma$ anticommutes with
$\gamma^\mu\partial_\mu\otimes\I_F$, as well as with 
$\gamma^\mu\omega_\mu\otimes\I_F$, as can be seen using the local form
 of the spin connection $\omega_\mu = \tfrac 14 \Gamma^{\nu \rho}_{\mu} \gamma_\rho\gamma_\nu$. 
 Hence:
\begin{equation}
   \label{eq:209}
   \left( \ds \otimes\I_F \right) \, \widetilde\Gamma + \widetilde\Gamma \left( \ds \otimes\I_F \right)=0.
 \end{equation}
On the other hand, the condition on the finite part $\gamma_{(2m)}\otimes D_F$, that is 
\begin{equation}
  \label{eq:210}
\left(\gamma_{(2m)}\otimes D_F \right) \, \widetilde\Gamma + 
\widetilde\Gamma \left(\gamma_{(2m)}\otimes D_F \right) = 0 , 
\end{equation}
follows from the zero-order and the twisted first-order conditions. For this one uses again a decomposition of 
the operator $J$ on the eigenbasis of $\widetilde\Gamma$; this being possible once requiring \eqref{eq:220}.  
\end{proof} 

\begin{rem}
\textup{
Usually the notion ``almost commutative geometry" is intended for the product of
the algebra of functions on a manifold by a finite dimensional noncommutative algebra. 
More generally, it could be used for any spectral triple where the
algebra $\A$ has an infinite dimensional center $Z(\A)$, while
$\A\slash Z(\A)$ is finite dimensional. A well known example which goes beyond the
product of a manifold by matrices is the noncommutative torus 
$(\A_\theta, \HH_\theta, D_\theta)$ spectral triple for a rational deformation parameter. In this case the algebra 
$\A_\theta$ is the algebra of endomorphisms of a bundle over a commutative torus and the center of $\A_\theta$ can be identified with the algebra of functions on this (usual) torus.    
Many of the results of the previous section extend to this more general cases, thus leading to other interesting examples. Details shall be reported elsewhere.}
\end{rem}

\section{Applications}
\label{sec:applications}

A twisted spectral triple for the Standard Model of elementary
particles has been proposed in
\cite{buckley}, whose twisted fluctuations yield the extra-scalar
field $\sigma$ required to stabilize the electroweak vacuum as pointed out in
\cite{Chamseddine:2012fk}, together with an
unexpected additional vector field $X_\mu$. It has been shown in
\cite{Martinetti:2014aa} that for $\M$ a four dimensional manifold, the appearance of $X_\mu$
is not due to the peculiar structure of the Standard Model, but is
a consequence of the twist on the commutative part of the almost
commutative geometry. 
  We generalize this result to any even dimensional manifold in Sect.~\ref{subsec_twistedfree} below. 
Then we study in Sect.~\ref{subsec:SM} to what extend the twisted spectral triple of \cite{buckley}
enters in the framework of minimal twisting introduced in the present paper.

\subsection{Twisted fluctuations of the free Dirac operator}
\label{subsec_twistedfree}

Let us consider the minimal twist of a even dimensional closed
Riemannian manifold $\M$ as described in
Proposition \ref{prop-dim-4}, that is
\begin{equation}
(\cinf\otimes \C^2, L^2(M, S),\ds; \rho) \quad\text{ where
}\quad \rho(f,g)=(g,f) \quad \forall f,g\in\cinf,
\label{eq:53}
\end{equation}
with grading $\gamma_{(2m)}$ and real structure $J$ (the `charge
conjugation' operator).
 
For the algebra $C^\infty(\M)$, the representation of the opposite algebra induced by $J$
is just the representation $\pi_\M$ composed with the involution, that is 
\begin{equation}
  \label{eq:29}
  {J}\pi_\M(f) {J}^{-1}= \pi_\M(\bar f).
\end{equation}
%  Since a commutative algebra can be identify with its opposite, for such an algebra one simply writes $J \pi(b^*) J^{-1} = \pi(b)$. 
% In particular, for the above minimal twists, when restricting to $f\in\cinf$, one identifies
% \begin{equation}
% {J}\pi_0(f) {J}^{-1}= \pi_0(\bar f).
% \label{eq:74}
% \end{equation}
A similar result holds for the minimal twist \eqref{eq:53}, but
depends on the $KO$-dimension.
%We denote by $\mbox{KO-dim}$ the $KO$-dimension of a triple.
\begin{lemma}
\label{lem:twistfluct}
%For the minimal twisted spectral triple in \eqref{eq:53} one has
  \begin{equation}
    \label{eq:67}
   {J}\pi(a) {J}^{-1}=
   \left\{\begin{array}{ll}
 \pi(a^*) &\text{if \,  KO-dim = 0, 4} \\ ~\\
           \pi(\rho(a^*))  &\text{if \, KO-dim = 2, 6}  
 \end{array}\right. .
 \end{equation}
\end{lemma}
\begin{proof}
The twisting automorphism in \eqref{eq:53} is such that $\rho^2=\id$.
Eq. \eqref{star-1}  then implies
\begin{equation}
\rho(a^*) = (\rho(a))^* . 
\label{eq:741}
\end{equation} 

For any $a= (f,g)\in\cinf\otimes\C^2$, Proposition \ref{prop-dim-4}
yields
\begin{equation}
  \label{eq:51}
  {J}\pi(a) {J}^{-1}= {J}\, p_+\pi_0(f) \,p_+\, {J}^{-1} + {J}\, p_-\pi_0(g)\, p_-\,{J}^{-1}
\end{equation} where $\pi_0=\pi_\M$ is the usual representation of $\cinf$ on
spinors and $p_{\pm}=\frac 12(\I\pm\gamma_{(2m)})$. 

If the KO-dimension is $0$ or $4$, the operator $J$ commutes with $\gamma_{(2m)}$,
hence with $p_+$ and $p_-$. Thus, using \eqref{eq:29}, 
\begin{align}
  \label{eq:73}
   {J}\pi(a) {J}^{-1} &= p_+ {J}\pi_0(f) {J}^{-1}
   p_+ + p_-{J}\pi_0(g) {J}^{-1} p_-   \nonumber \\
&= p_+ \pi_0(\bar f) +
   p_-\pi_0(\bar g) = \pi(\bar f, \bar g) = \pi(a^*).
\end{align}
In KO-dimension $2$ or $6$, the operator $J$ anticommutes with
$\gamma_{(2m)}$, meaning that ${J} p_+= p_- {J}$ and  ${J} p_-= p_+ {J}$. Hence, using now \eqref{eq:741}, 
\begin{align}
  \label{eq:7003}
   {J}\pi(a) {J}^{-1} &= p_- {J}\pi_0(f) {J}^{-1}
   p_- + p_+{J}\pi_0(g) {J}^{-1} p_+ \nonumber \\
&= p_- \pi_0(\bar f) +
   p_+\pi_0(\bar g) = \pi(\bar g, \bar f) = \pi(\rho(a^*)).
\end{align}
Thus the statement \eqref{eq:67}. 
\end{proof}

Now, if $\text{dim}\, \M = 2m$, any $a=(f, g)\in\cinf\otimes\C^2$, one has 
\begin{equation}
  \label{eq:6}
  \pi(a) = \left(\begin{array}{cc} f\I_{2^{m-1}}& 0 \\0 & g
      \I_{2^{m-1}}\end{array}\right), \qquad \pi(\rho(a)) = \left(\begin{array}{cc} g\I_{2^{m-1}}& 0 \\0 & f
      \I_{2^{m-1}}\end{array}\right).
\end{equation}
Using the fact %\rosso{\cite{Walter}} 
that the spin connection commutes with the representation 
(and omitting the symbol of representation) a direct computation leads to  
\begin{align}
  \label{eq:40}
  [\ds,  a]_\rho & = -\ii \gamma^\mu[\partial_\mu, a] + 
  \big(\gamma^\mu \, a - \rho(a)\,\gamma^\mu\big) \, \ds \nonumber 
  \\ & = -\ii \gamma^\mu (\partial_\mu a),
\end{align}
since from Lemma~\ref{commutator-Dirac-matrices} for the particular automorphism $\rho$ in \eqref{eq:6} 
one has 
\begin{align}
\gamma^\mu a = \rho(a)\,\gamma^\mu .
\label{eq:22}
\end{align}
Using again this, any twisted $1$-form as defined in \eqref{eq:83} can thus be written as 
\begin{equation}
  \label{eq:45}
  A_\rho = -\ii \sum\nolimits_j a_j\,\gamma^\mu (\partial_\mu b_j) =: -\ii \gamma^\mu
  \sum\nolimits_j \rho(a_j) (\partial_\mu b_j)  \quad\quad \mbox{for} \quad a_j, b_j\in\cinf\otimes\C^2 . 
\end{equation}

\begin{lemma}\label{lemma5.2}
For the minimal twisted spectral triple in \eqref{eq:53} one has
  \begin{equation}
    \label{eq:451}
   {J} A_\rho {J}^{-1}=
   \left\{\begin{array}{ll}
 - \rho(A_\rho^*)  &\text{if KO-dim = 0, 4} \\
 ~\\
  - A_\rho^*  &\text{if KO-dim = 2, 6}  
 \end{array}\right. .
 \end{equation}
\end{lemma}
\begin{proof}
In even dimensions the real structure $J$ commutes with the Dirac operator, $JD = DJ$ so that from the signs in 
\eqref{eq:34} one has $\epsilon'=1$. Being $J$ antilinear this means that 
\begin{equation}
  \label{eq:56}
  J\gamma^\mu = - \gamma^\mu J
\end{equation}
since usual gamma matrice algebra yields that $J$ commutes with the
covariant spin derivatives $\nabla_\mu$. 
By Lemma \ref{lem:twistfluct}, since $J$ is antilinear, it commutes with
$\partial_\mu$ and $\rho$ is a $*$-automorphism from \eqref{eq:741}, direct computations yields
\begin{equation}
  \label{eq:861}
  J A_\rho J^{-1} =\left\{\begin{array}{ll} 
  -\ii \gamma^\mu\sum_j \rho(a_j^*) (\partial_\mu b_j^*) & \text{ if KO-dim = 0,4} \\ ~\\ 
  -\ii \gamma^\mu \sum_j  a_j^* (\partial_\mu \rho(b_j^*)) & \text{ if KO-dim = 2,6 } 
\end{array}\right. . \end{equation}
On the other hand, using \eqref{eq:741} and \eqref{eq:22}, one computes:
\begin{equation}
\label{eq:24}
 A_\rho^* =  \ii \sum\nolimits_j (\partial_\mu b_j^*) \rho(a_j^*) \, \gamma^\mu = 
 \ii \gamma^\mu \sum\nolimits_j (\partial_\mu \rho(b_j^*)) a_j^* = 
 \ii \gamma^\mu \sum\nolimits_j a_j^* (\partial_\mu \rho(b_j^*)) , 
\end{equation}
since $(\partial_\mu \rho(b_j^*)) \in\cinf$ commutes with $a_j^* \in\cinf$. 
With a slight abuse of notation due to the omission of the symbol of representation, we denote the first line of  the r.h.s. of \eqref{eq:861} as $\rho(A_\rho^*)$. The results in \eqref{eq:451} follows by comparison.
\end{proof}

\begin{prop}
  \label{prop:fluctfree}
  There are no twisted fluctuations of the Dirac operator $\ds$ if the KO-dimension is $2$ or $6$. On the other hand, for KO-dimension $0$ or $4$, the twisted fluctuations are of the form
  %  \begin{equation}
%     \label{eq:52}
%     \ds_\rho = \ds + A_\rho - \rho(A_\rho^*)   
%   \end{equation}
% where  $A_\rho$ is any  element of $\Omega_D^1$ such that $A_\rho - \rho(A_\rho^*)$ is self-adjoint.
    \begin{equation}
      \label{eq:63}
      \ds_\rho = \ds - \ii \gamma^\mu \, f_\mu \gamma_{(2m)},
    \end{equation}
where $f_\mu = (f_1, \dots, f_{2m})$ are arbitrary real functions in $\cinf$. 
\end{prop}
\begin{proof}
From Lemma~\ref{lemma5.2} one has 
\begin{equation}
   \label{eq:891}
   \ds_\rho = \ds + A_\rho + {J} A_\rho {J}^{-1} = \ds + A_\rho - 
    \left\{\begin{array}{ll}
 \rho(A_\rho^*)  &\text{if KO-dim = 0, 4} \\
 ~\\
  A_\rho^*  &\text{if KO-dim = 2, 6}  
\end{array}\right. . 
\end{equation}
By requiring that $\ds_\rho$ be self-adjoint one sees that for the
KO-dimension $2$ or $6$, the additional term $A_\rho-A_\rho^*$
equals its opposite, hence it vanishes.
For KO-dimension $0$ or $4$, let us write 
\begin{equation}
  \label{eq:68}
  Y_\mu := \sum\nolimits_j \rho(a_j)\, (\partial_\mu b_j), \qquad
  \rho(Y_\mu) :=\sum\nolimits_j a_j\, (\partial_\mu \rho(b_j)) 
\end{equation}
so that \eqref{eq:45} and \eqref{eq:24} yields
\begin{equation}
  \label{eq:72}
  A_\rho = -\ii \gamma^\mu Y_\mu,\quad A_\rho^* = \ii \gamma^\mu Y_\mu^*.
\end{equation}
Therefore $\ds$ is selfadjoint if and only if
 \begin{equation}
A_\rho - \rho(A_\rho^*) = -\ii \gamma^\mu(Y_\mu + Y_\mu^*) 
%=-\ii \gamma^\mu \, \text{diag}(f_\mu + f_\mu^*, g_\mu + g_\mu^*)\otimes\I_{2^{m-1}} .
 \label{eq:70}
 \end{equation}
 is self-adjoint. By \eqref{eq:22} this is equivalent to
\begin{equation}
  \label{eq:243}
\gamma^\mu\left(\rho(Y_\mu +Y_\mu^*) + Y_\mu + Y_\mu^*\right)=0.
\end{equation}
With $a_j=(f_j, g_j)$ and  $b_j=(f'_j, g'_j)$ in $\cinf\otimes \C^2$, one has 
 \begin{equation}
Y_\mu  :=
\text{diag}\left(f_\mu \,\I_{2^{m-1}},\,  g_\mu \,\I_{2^{m-1}}\right),
\quad \rho(Y_\mu)  :=
\text{diag}\left(g_\mu \,\I_{2^{m-1}},\,  f_\mu \,\I_{2^{m-1}}\right),
\label{eq:9300}
\end{equation} 
where  
$f_\mu:=\sum\nolimits_j\, g_j \,\partial_\mu f'_j$ and $g_\mu:=\sum\nolimits_j\, f_j\, \partial_\mu g'_j$.
Both $Y_\mu + Y_\mu^*$ and $\rho(Y_\mu +Y_\mu^*)$ are block
diagonal matrices with block $\cinf$-proportional to $\I_{2^{m-1}}$, so  the
  l.h.s. of \eqref{eq:243} is block off-diagonal, with blocks $\cinf$-linear
  combinations of Pauli matrices. Hence \eqref{eq:9300} is equivalent
  to 
\begin{equation}
  \label{eq:48}
  Y_\mu +Y_\mu^* = -\rho(Y_\mu + Y_\mu^*).
\end{equation}
This means \begin{equation}
g_\mu + g_\mu^* = -(f_\mu + f_\mu^*)
\end{equation}
which is the same as 
\begin{equation}
  \label{eq:9201}
Y_\mu + Y_\mu^* = 2(\text{Re}\, f_\mu)\,\gamma_{(2m)} . 
\end{equation}
The latter is of the form in \eqref{eq:63}. This concludes the proof.
\end{proof}

In the non-twisted case, that is when $\rho$ the identity automorphism,
then \eqref{eq:891} shows that the fluctuations of $\ds$ also vanish
in $KO$-dimension $0, 4$.  This can also be read in \eqref{eq:48}, which for
$\rho=\text{Id}$ implies that $Y_\mu
+Y_\mu^*$ equals its opposite, hence is zero. One retrieves the well
known result that (non-twisted) fluctuations of the Dirac operator in
the commutative case always vanish. 
 
%\subsection{The Standard Model} 
\subsection{On twisting the spectral Standard Model} 
\label{subsec:SM}
We investigate how the twisted spectral triple for the Standard Model of elementary
particles proposed in \cite{buckley}  fits the framework of
the present paper.

The (non-twisted) spectral triple of the Standard Model
\cite{Chamseddine:2007oz} is 
the almost commutative geometry
\begin{equation}
\A=\cinf\otimes\A_{sm},\quad\HH= L^2(\M, S)\otimes\HH_F,\quad   D=
\ds\otimes\I_F + \gamma_{(2m)}\otimes D_F 
\label{eq:240}
\end{equation}
where
\begin{equation}
  \label{eq:149}
  \A_{sm}:= \C\oplus {\mathbb H}\oplus \IM_3(\C) 
\end{equation} 
acts on the finite dimensional space $\HH_F$ whose dimension is the number of
elementary fermions. Then $D_F$ is a matrix acting on $\HH_F$ whose
coefficients encode the masses of these fermions. As in
\cite{buckley} we work with one generation only, so that $\HH_F\simeq \C^{32}$ splits as
\begin{equation}
  \label{eq:237}
  \HH_F =  \HH_L \oplus  \HH_R \oplus \HH_L^a \oplus \HH_R^a
\end{equation}
with each of the summands isomorphic to $\C^8$
($8$ is for one pair of colored quarks and one pair electron/neutrino). The index $L/R$ is for
left/right particles, and the exponent $a$ is for antiparticles.
The (real) algebra of quaternion acts only on $\HH_L$, the algebra $\IM_3(\C)$ only on
$\HH_L^a \oplus \HH_R^a$ and $\C$ on  $\HH_R \oplus \HH_L^a \oplus
\HH_R^a $, namely for $c\in\C, q\in\mathbb H$ and $m\in\IM_3(\C)$ one has
\begin{equation}
  \label{eq:140}
  \pi_F(c, q, m) =  \pi_L(q) \oplus \pi_R(c) \oplus \pi_L^a(c, m) \oplus
  \pi_R^a(c, m).
  \end{equation}
Explicitly, identifying a quaternion $q$ with its usual representation as
$2\times 2$ complex matrix, one has 
\begin{align}
  \label{eq:238}
 & \pi_L(q) :=
  q\otimes \I_4, \quad  \pi_R(c) := \text{diag} (c, \bar c)\otimes \I_4, \nonumber \\[4pt] 
 & \pi^c_L(c,m)= \pi^c_R(c,m) := \I_2\otimes \text{diag}  (c,m).
\end{align}
The identity $\I_4$ in the particle sector means that $\C$ and $\mathbb H$
preserve the color, and do not mix leptons with quarks. The identity $\I_2$ in the antiparticle sector means
that $\C$ and $\IM_3(\C)$ preserves the flavour: $c$ acts by multiplication on
antileptons while $\IM_3(\C)$ mixes the color of the antiquarks. The
representation of $\A$ on $\HH$ is thus
\begin{equation}
  \label{eq:239}
\pi_0(f\otimes a_F) = \pi_\M(f)\otimes \pi_F(a_F) \quad \forall
f\in\cinf, \, a_F\in\A_{sm}.   
\end{equation}

A twisted spectral triple $(\widetilde\A, \HH, D; \rho)$  of the
Standard Model has been 
obtained in \cite{buckley} following an idea introduced in \cite{Devastato:2013fk}.
One lets the algebra $\C\oplus{\mathbb H}$ act independently on the
left/right components of spinors, only in the particle sector $\HH_L \oplus  \HH_R $ , that is $\C\oplus{\mathbb H}$ is doubled but 
$\IM_3(\C)$ is not. 
Explicitly one takes $\widetilde\A = \cinf \otimes \widetilde\A_{sm}$ where 
\begin{equation}
  \label{eq:202}
  \widetilde\A_{sm}: =\C\oplus \C \oplus {\mathbb H}\oplus {\mathbb H}
  \oplus \IM_3(\C), 
\end{equation}
This partial doubling can be easily dealt with by extending our Definition~\ref{deftwist} of a minimal twist. 
\begin{df}
\label{def:partial}
Let $(\A, \HH, D)$  be a spectral triple whose algebra
\begin{equation}
  \label{eq:69}
  \A=\A' \oplus\A''
\end{equation}
is the direct sum of two (pre-) $C^*$ algebras $\A'$ and $\A''$. 
A generalised minimal twist of $(\A, \HH, D)$ by the algebra $\B$ 
  is a twisted spectral triple
  \begin{equation}
((\A'\otimes\B) \oplus \A'',\HH,D; \rho)
\label{eq:234}
\end{equation}
such  that the initial representation $\pi_0$ of $\A'\oplus\A''$ on
$\HH$ is
retrieved from the representation $\pi$ of the algebra $(\A'\otimes\B)\oplus\A''$ as
  \begin{equation}
    \label{eq:150}
    \pi_0(a'\oplus a'') = \pi \big((a'\otimes\I_\B )\oplus a'' \big) \quad \forall
    a'\in\A',\, a''\in\A'',\, b\in\B.
 \end{equation}
\end{df}
\noindent We could have taken from the very beginning this more general definition rather than the one in Definition~\ref{deftwist}. We have decided not to do so, since 
this would have only made the paper rather cumbersome and heavier to read while not adding much to its scientific content. 

In the case of the twisted spectral triple for the Standard Model, by setting $\A'=\cinf\otimes \C\oplus {\mathbb H}$ and $\A''=\cinf\otimes\IM_3(\C)$ so that
\begin{equation}
\cinf\otimes\A_{sm}= \A' \oplus \A'',
\label{eq:79}
\end{equation}
one gets as expected
\begin{equation}
  \label{eq:245}
  \cinf\otimes \widetilde\A_{sm} =(\A'\otimes \B) \oplus \A''
\end{equation}
with twisting algebra $\B =\R^2$ --- one cannot consider $\B=\C^2$, for $\mathbb H$
is not a complex algebra.  The representations $\pi$ of $\cinf\otimes\widetilde\A_{sm}$ that, together with the initial representation $\pi_0$ in \eqref{eq:239} satisfies \eqref{eq:150}, is given by
\begin{multline}
  \label{eq:236}
  \pi(f\otimes A) :=  
 \left(p_+ \pi_\M(f) \right) \otimes \Big(\pi_L(q^r) + \pi_R(c^r)\Big) 
 \;+\; 
 \left(p_- \pi_\M(f) \right) \otimes \Big( \pi_L(q^l) + \pi_R(c^l) \Big) \\[4pt]  + 
\pi_M(f)\otimes \Big(\pi_L^a(c^r, m) + \pi_R^a(c^r, m)\Big) . \qquad
\end{multline}
Here $p_\pm:= \tfrac 12 (\I_\M\pm \gamma_{(2m)})$ and the 
generic element of the algebra $\widetilde\A_{sm}$ in \eqref{eq:202} is 
\begin{equation}\label{algw}
  A= (c^r, c^l, q^r, q^l,m) \quad \text{ with }\quad (c^r, c^l)\in\C^2, \;
  (q^r, q^l)\in{\mathbb H}^2, \; m\in \IM_3(\C) .
\end{equation}
 
In contrast with the construction of the present paper, 
the automorphism $\rho$ of the twisted spectral triple of \cite{buckley} is 
an automorphism of the represented algebra $\pi(\widetilde\A_{sm})$ rather than  
$\widetilde\A_{sm}$ itself. With the notation \eqref{algw}, 
this automorphism exchanges $(q^l, c^l)$ with $(q^r, c^r)$ in the particle sector, 
while leaving unchanged the $c^r$ in the anti-particle sector. Explicitly, 
\begin{multline}
  \rho\left(\pi(f\otimes A)\right) =  
\left(p_+ \pi_\M(f)\right) \otimes \Big(\pi_L(q^l) +\pi_R(c^l)\Big) \;+\; 
p_- \pi_\M(f)\otimes \Big( \pi_L(q^r) + \pi_R(c^r)\Big)\\[4pt]  + 
\pi_M(f)\otimes \Big(\pi_L^a(c^r, m) + \pi_R^a(c^r, m)\Big) . \qquad
\label{eq:244}
  \end{multline}

Additional investigation on this point will be reported elsewhere. 
One
option is to generalise the results of the present paper to automorphisms that do not
commute with the representation, so as to fit the twisted spectral
triple of \cite{buckley} in the scheme. A second possibility is to 
minimally twist the Standard Model in the sense of 
Definition \ref{deftwist} or Definition \ref{def:partial}, and see whether twisted
fluctuations still generate the extra-scalar field $\sigma$ needed for the model, or even more general fields.

That the twisted spectral triple of \cite{buckley} does not completely fit our main definition is a sign that there could be  more general models for twisted spectral triples for the Standard Moled of particle physics, 
leading hopefully to richer phenomenological consequences.

%\bigskip
%\noindent
%-------- -------- -------- -------- -------- \\
%\noindent
%Giovanni Landi, Universit\`{a} di Trieste, Trieste, Italy and I.N.F.N. Sezione di Trieste, Trieste, Italy. 
%\noindent
%Pierre Martinetti, Universit\`{a} di Trieste, Trieste and Universit\`{a} di Genova, Genova, Italy. \\
%emails: landi@units.it , martinetti@dima.unige.it

\end{document}